%% file: main.tex
\documentclass[11pt,a4paper]{article}

\usepackage[margin=1in]{geometry}
\usepackage[utf8]{inputenc}
\usepackage[T1]{fontenc}
\usepackage{amsmath,amssymb,amsthm,mathtools,mathrsfs}
\usepackage{tikz}
\usetikzlibrary{arrows.meta}
\usepackage{xcolor}
\definecolor{diagramblue}{RGB}{35,80,100}
\usepackage{graphicx}
\graphicspath{{figures/}}
\usepackage{booktabs,array,longtable}
\usepackage{siunitx}
\usepackage[numbers,square,sort&compress]{natbib}
\usepackage[hidelinks]{hyperref}
\hypersetup{
    pdftitle={Associator Defects in Non-Unitary String-Nets},
    pdfauthor={Hanshi Yang}
}
\usepackage[nameinlink,noabbrev]{cleveref}

\numberwithin{equation}{section}

\newtheorem{theorem}{Theorem}[section]
\newtheorem{lemma}[theorem]{Lemma}
\newtheorem{corollary}[theorem]{Corollary}
\newtheorem{proposition}[theorem]{Proposition}

\theoremstyle{remark}

\theoremstyle{definition}
\newtheorem{question}[theorem]{Question}

\newcommand{\manualstatement}[3]{\par\addvspace{6pt}\refstepcounter{theorem}\label{#3}\noindent\textbf{#1~\thetheorem\if\relax\detokenize{#2}\relax\else\ (#2)\fi.}\quad}

\newcommand{\C}{\mathcal C}
\newcommand{\K}{\mathcal K}
\newcommand{\Hh}{\mathcal H}
\newcommand{\one}{\mathbf1}
\newcommand{\id}{\mathrm{id}}
\newcommand{\Tr}{\operatorname{Tr}}
\newcommand{\Def}{D_{\mathrm{assoc}}}
\newcommand{\norm}[1]{\left\lVert#1\right\rVert}
\newcommand{\eps}{\varepsilon}
\newcommand{\spec}{\operatorname{spec}}
\newcommand{\ran}{\operatorname{ran}}

\begin{document}

\vspace{1.5cm}

\begin{center}
\LARGE\bfseries
Associator Defects in Non-Unitary String-Nets
\end{center}

\vspace{1cm}

{\large Hanshi Yang $^{1}$}

\vspace{0.5cm}
\noindent
$^1$ {\it Northwood High School}\\
\vspace{0.1cm}
{\tt \footnotesize mainhanshiyangemail@gmail.com}

\vspace{1cm}

\begin{abstract}
\normalsize
\noindent
We study obstructions to realizing fusion categories in physical Hilbert spaces by varying positive inner products on elementary fusion spaces. We introduce an associator defect to measure the smallest failure of recoupling maps to preserve the induced inner products in order to do so. By using a determinant identity from the pentagon equation, we prove that the minimum is attained and vanishes when the category is unitarizable. As an example, we compute the defect for the Yang-Lee category and its double, studying its behavior under stacking and passage to Drinfeld centers, which results in a Morita-invariant measure of bulk nonunitarity. We also obtain lower bounds for quantum-group categories with fusion multiplicities. Finally, we study how these obstructions appear in physical realizations of fusion. Under the stated assumptions on local splitting operations and coherent measurements, a positive defect imposes quantitative tradeoffs between implementation errors and operator norms. We also obtain bounds on the response to local perturbations. These results show how the failure of categorical unitarity will place quantitative constraints on physical realizations of potential fusion.
\end{abstract}

\clearpage
\tableofcontents
\clearpage

\input{01-introduction}

\input{02-fusion-metrics}

\input{03-products-centers}
\input{04-local-fusion}

\input{05-measurements}

\clearpage
\appendix
\input{06-appendix-fusion-bases}

\input{07-appendix-coupled-descent}

\input{08-appendix-yang-lee}

\input{09-appendix-dimensions}

\input{10-appendix-physical-lifts}

\bibliographystyle{alpha}
\bibliography{references}

\end{document}

%% file: 01-introduction.tex
\section{Introduction}\label{sec:intro}

A fusion category specifies how to combine charges and how to pass between different ways of grouping successive fusions. It is the pentagon equation that makes these changes consistent, but that says little about their compatibility with a positive inner product.

A string-net construction \cite{LW} starts with fusion and some recoupling data, but for it to be physical, we must also require an inner product.

In the unitary case, this compatibility is built into the Levin--Wen construction which gives a Hermitian commuting-projector Hamiltonian \cite[Section IV and Appendix C]{LW}.

If we look at a recoupling map, one could simply take an inner product to be transported to the target. However, if we repeat this independently for every map, we miss something incredibly interesting.
A fusion vertex is reused in different trees, with different external charges and spectators, and its form must be the same in each occurrence: therefore, a change that might improve one recoupling can spoil another.

More precisely, we can assign a positive-definite Hermitian form, which we call an elementary metric, to each elementary fusion space. These forms induce inner products on composite fusion spaces through tensor products and orthogonal direct sums. Now it is pretty common to wonder whether one choice of elementary metrics can make every recoupling map an isometry simultaneously.

A simple example of this is an idempotent on $\mathbb C^2$,
\[    P=\begin{pmatrix}1&t\\0&0\end{pmatrix}.\]

One can quickly see that it satisfies $P^2=P$ for every $t$, yet in the standard inner product is only an orthogonal projection for when $t=0$. A different positive inner product can make this particular idempotent orthogonal, however. The question in a fusion problem is whether one choice to do so compatibly for the whole collection of operations knowing that the forms on elementary vertices are shared across their different occurrences.

For some real $t$, the idempotent is self-adjoint in the positive metric:
\[
    G_t=\begin{pmatrix}1&t\\t&1+t^2\end{pmatrix},
    \qquad \det G_t=1,
    \qquad G_tP=P^\dagger G_t.
\]

We can see that positivity follows from the positive upper-left entry and positive determinant, with the range spanned by $(1,0)$ and its kernel by $(-t,1)$. The $G_t$ inner product is zero. Of course this does mean that the failure of self-adjointness in the original coordinates is easily repairable for this particular operator.
We do, however, wish to study cases where the failure to make all the required repairs with one coherent collection of elementary forms.

There are two main questions that we would want to answer.
\vspace{2mm}

First, we have a mathematical question: how closely can a category's fusion and recoupling maps be made compatible with positive forms?

Secondly, how does this translate into the Hilbert-space structure of a physical system? Specifically, which overlaps of states or local operations would require such a compatible family of positive forms to exist?

The solution to the first question can be solved by varying elementary metrics, and the second can be solved by reconstructing those metrics from physical operations.

\subsection{Inner products on fusion spaces}

Let $\mathcal C$ be a complex fusion category: a finite semisimple rigid linear monoidal category with simple tensor unit. Its simple objects will be called charges; their Frobenius--Perron dimensions and the global Frobenius--Perron dimension are
\begin{equation}\label{intro:FP}
    d_a=\operatorname{FPdim}(a),\qquad
    \Phi=\operatorname{FPdim}(\mathcal C)=\sum_a d_a^2.
\end{equation}

See \cite[Sections 3.3 and 4.1]{EGNO} for these dimension functions. The elementary fusion space $V_{ab}^c=\operatorname{Hom}(c,a\otimes b)$ notes the ways to split $c$ into $a,b$. For three inputs, the associator relates the two orders of splitting shown in Figure~\ref{fig:fmove}:

\begin{equation}\label{intro:recoupling}
    \bigoplus_e V_{ab}^{e}\otimes V_{ec}^{d}
    \xrightarrow{\ F_{abc}^{d}\ }
    \bigoplus_f V_{bc}^{f}\otimes V_{af}^{d},
    \qquad V_{ab}^{c}=\operatorname{Hom}(c,a\otimes b).
\end{equation}

Furthermore, a positive form $g_{ab}^c$ at each elementary vertex gives
\begin{equation}\label{intro:treemetrics}
    G_L=\bigoplus_e g_{ab}^e\otimes g_{ec}^d,\qquad
    G_R=\bigoplus_f g_{bc}^f\otimes g_{af}^d.
\end{equation}

Yet different intermediate charges are orthogonal in these forms. Recoupling would preserve them when

\begin{equation}\label{intro:metric}
    F^\dagger G_RF=G_L.
\end{equation}

The same $g_{ab}^c$ occurs wherever that vertex appears, including in the presence of different spectators. With this requirement, solving \eqref{intro:metric} for all blocks is equivalent to unitarizing $\mathcal C$ \cite[Remark 3.14(iii)]{FG}.

\begin{figure}[tbp]
\centering
\begin{tikzpicture}[x=0.8cm,y=0.8cm,line width=0.7pt]
\draw (-4,2)--(-3,0.9)--(-2,0)--(-2,-1.2);
\draw (-2.8,2)--(-3,0.9);
\draw (-0.8,2)--(-2,0);
\node[above] at (-4,2) {$a$};
\node[above] at (-2.8,2) {$b$};
\node[above] at (-0.8,2) {$c$};
\node[left] at (-2.55,0.5) {$e$};
\node[below] at (-2,-1.2) {$d$};
\draw (0.8,2)--(2,0)--(2,-1.2);
\draw (2.8,2)--(3,0.9)--(2,0);
\draw (4,2)--(3,0.9);
\node[above] at (0.8,2) {$a$};
\node[above] at (2.8,2) {$b$};
\node[above] at (4,2) {$c$};
\node[right] at (2.55,0.5) {$f$};
\node[below] at (2,-1.2) {$d$};
\draw[->] (-0.65,0.3)--(0.65,0.3);
\node[above] at (0,0.4) {$F_{abc}^d$};
\end{tikzpicture}
\caption{The associator changes the fusion tree while preserving the external charges. The metrics on both trees are induced from the same elementary vertex forms.}
\label{fig:fmove}
\end{figure}

However, there is a normalization issue in measuring the failure of \eqref{intro:metric}. If the forms are allowed to change their scales, the size of the matrix difference can no longer measure the change in lengths. We therefore have to normalize by the input form and try to minimize the relative error. With the unit channels fixed, we will use the associator defect that we define:

\begin{equation}\label{intro:defect}
    D_{\mathrm{assoc}}(\mathcal C)
    =\inf_{g>0}\max_B\|U_B(g)^\dagger U_B(g)-I\|,
    \qquad U_B(g)=G_{R,B}^{1/2}F_BG_{L,B}^{-1/2}.
\end{equation}

$B=(a,b,c;d)$ ranges over nonempty associator blocks. The matrix for the norm is the largest relative error in squared length. We can minimize over the elementary forms to remove dependence on fusion bases as such

\[
    G_{L,B}^{-1/2}
    \bigl(F_B^\dagger G_{R,B}F_B-G_{L,B}\bigr)G_{L,B}^{-1/2}.
\]

The infimum is a tad interesting to look at here. Failure of unitarizability rules out any exact solution, but does not necessarily rule out an increasingly accurate solution whose elementary forms degenerate, which of course would also be worth looking at. In fact, the metric variables need not stay bounded even when the normalized associators are fixed.
There is a second difficulty at defect levels above one: a singular value can approach zero without making the Gram error large. We would therefore want to introduce some compactness argument too to control the inverse associators.

The relation that supplies this control comes from the pentagon. Its absolute determinant, summed over labels with Frobenius--Perron weights, gives
\begin{equation}\label{intro:balance}
    \sum_{a,b,c,d}d_ad_bd_cd_d\log|\det F_{abc}^{d}|=0.
\end{equation}

When one expresses each logarithmic determinant as a sum of logarithmic singular values, under a common upper bound on the forward norms, it is clear that every term has an upper bound. Now, since the weights are positive and the weighted sum is zero, no individual term can tend towards minus infinity, meaning consequently the inverse norms are bounded. When one passes to a limit now invertibility and the pentagon are preserved.
Now, onto duality. In Lemma~\ref{fusion:rigidlimit} we can normalize evaluation and coevaluation maps using the one-dimensional vacuum channels and show that these maps converge. After establishing the rigidity of the limit, we can use the openness of its gauge orbit supplied by Ocneanu rigidity \cite[Theorem 9.1.4]{EGNO}. This naturally proves attainment and, in particular, that
\begin{equation}\label{intro:zero}
    D_{\mathrm{assoc}}(\mathcal C)=0
    \quad\Longleftrightarrow\quad
    \mathcal C\ \text{is unitarizable}.
\end{equation}

Unitarizable means tensor equivalent to a unitary fusion category, in which the morphism spaces carry positive-definite Hermitian inner products compatible with composition and tensor products.

The compact object in the proof is a bounded family of normalized associators introduced. There is no corresponding compactness assertion for all choices of elementary metrics however; Section~\ref{sec:direct} gives a simple family of gauges that diverges while leaving every associator unchanged.

Furthermore, even when a candidate metric could be found explicitly, checking if it is the optimal one is a separate, difficult question. For example, a calculation in a convenient family of metrics may leave many elementary variations that have not necessarily been explored. Moreover, the maximum block norm is generally not differentiable when several blocks meet.

We address these difficulties by reformulating the metric constraints such as  as geodesically convex inequalities to attempt to establish global optimality.

Symmetries can reduce the number of variations in some examples.  A finite symmetry has an invariant minimizing metric, and for the canonical permutation action on Deligne factors these forms become equivariant objects. In particular,
\begin{equation}\label{intro:equiv}
    D_{\mathrm{assoc}}((\mathcal C^{\boxtimes n})^G)
    \le D_{\mathrm{assoc}}(\mathcal C^{\boxtimes n}),
    \qquad G\le S_n,
\end{equation}

\subsection{Products and bulk charges}
Let us take a look at the fusion rule $\tau^2=1+\tau$.

Fibonacci is unitary; Yang--Lee is not unitarizable \cite[Main Theorem and Section 2.5]{Ostrik}. Both their associators are Galois conjugate \cite[Eqs. (6)--(7)]{Lootens}. In the symmetric Yang--Lee gauge, the only nonscalar block is
\begin{equation}\label{intro:YLmatrix}
    F_{\tau}^{\tau\tau\tau}
    =\begin{pmatrix}-\phi&i\sqrt\phi\\i\sqrt\phi&\phi\end{pmatrix},
    \qquad \phi=\frac{1+\sqrt5}{2}.
\end{equation}

Both trees are diagonal, and their relative channel weights are the only variable seen by this block.

We can apply some of the theory discussed above.

\begin{theorem}\label{intro:YLtheorem}
The Yang--Lee category has associator defect
\begin{equation}\label{intro:YLsingle}
    D_{\mathrm{assoc}}(\mathrm{YL})=(\phi+\sqrt\phi)^2-1.
\end{equation}
The minimum is attained when $g_{\tau\tau}^{1}=(g_{\tau\tau}^{\tau})^2$.
\end{theorem}

This double is also the category of bulk charges of a Yang--Lee string net. Generally, bulk charges for input $\mathcal A$ belong to $Z(\mathcal A)$: furthermore, the extra half-braiding specifies how exactly an excitation passes an input \cite[Theorems 6.4 and 7.3]{Kirillov}\cite[Sections 4 and 5]{MPO}.

Hence, set
\begin{equation}\label{intro:bulk}
    D_{\mathrm{bulk}}(\mathcal A)
    =D_{\mathrm{assoc}}(Z(\mathcal A)).
\end{equation}

Because Morita equivalent inputs have equivalent centers \cite[Theorem 8.12.3]{EGNO}, our definition depends only on the bulk fusion structure. For Yang-Lee, the center is monoidally equivalent to two Yang-Lee factors \cite[Proposition 8.20.12]{EGNO}, where their coupled minimum is
\begin{equation}\label{intro:YL}
    D_{\mathrm{bulk}}(\mathrm{YL})=\lambda^2-1,
    \qquad
    \lambda+\lambda^{-1}=3\phi+1,\quad \lambda>1,\quad
    \phi=\frac{1+\sqrt5}{2}.
\end{equation}

For general products, minimizing within the individual factors need not establish optimality in the product category, where additional elementary metric variations can become available. Section~\ref{sec:coupled} gives a condition for optimality and estimates the improvement when it fails.

Furthermore, we obtain lower bounds without using the full
associator matrices. In vacuum channels, the relevant elementary fusion spaces are instead one dimensional, allowing us to track their gauge factors directly.

For spherical input, one can put $r(\mathcal A)=\operatorname{FPdim}(\mathcal A)/\dim(\mathcal A)$. The vacuum amplification in the center is therefore $r(\mathcal A)^2$. When $r(\mathcal A)>1$, this gives
\begin{equation}\label{intro:stacking}
    D_{\mathrm{bulk}}(\mathcal A^{\boxtimes n})
    >r(\mathcal A)^{2n}-1,\qquad n\ge1.
\end{equation}

This would allow arbitrary---and potentially even coupled---metrics on the product category.

Now, for a self-dual simple of intrinsic squared norm  $v^{-1}<1$, we obtain a stronger bound. If $J$ is its inverse-norm exponent, then $1+D_{\mathrm{assoc}}\ge s_*>v$, where \[(s_*-v)(s_*^J-v)=(v-1)^2.\]

This applies, in particular, to the family $\mathcal Q_L$ of Galois conjugates of $SU(3)_{L-3}$ with odd $L\ge7$, whose adjoint occurs twice in its tensor square. The associated three-fold recoupling block is eight dimensional with the multiplicity form as an unrestricted positive $2\times2$ matrix. Nevertheless, their vacuum coefficients alone can give
\begin{equation}\label{intro:quantum}
    1+D_{\mathrm{assoc}}(\mathcal Q_L)
    >\frac{1}{16\cos^2(\pi/L)(1-\cos(\pi/L))^2}
    >\frac{L^4}{4\pi^4}.
\end{equation}

Dimension bounds have a natural limit. In the pseudo-unitary case $r=1$, a canonical positive spherical structure exists \cite[Proposition 9.5.1]{EGNO}, but it is not known currently whether it always extends to a unitary fusion structure \cite[Remark 6.6]{GP}.

In our notation introduced, the question is
\begin{equation}\label{end:open}
    r(\mathcal C)=1
    \quad\stackrel{?}{\Longrightarrow}\quad
    D_{\mathrm{bulk}}(\mathcal C)=0.
\end{equation}

\subsection{Charged operations in a lattice model}
The preceding questions concern fusion data. When considering a lattice, we require a tad more information: which charged spaces represent these data, and which operators implement fusion on them? Unfortunately, the spherical graphical calculus alone cannot answer this question. For example, it defines puncture projectors \cite[Theorem 3.9 and Section 7]{Kirillov}, and in the unrescaled loop convention the vacuum plaquette operator takes the form
\begin{equation}\label{intro:plaquette}
    B_p=\frac{1}{\Delta}\sum_s q_sL_p^s,\qquad
    \Delta=\sum_s q_s^2,
\end{equation}
for spherical dimension $q_s$. Fusion and duality further gives $B_p^2=B_p$. Whether it is self-adjoint is a further property of the physical realization.

One way we can obtain the elementary forms is to examine the physical overlaps of split states. We consider splitting operators $t_{ab,s}^c(v)$, where $v$ is a fusion vector and $s$ denotes the spectator context.

Now, write $q_\ell$ for the charged scalar-compression error, $K_\ell$ for the implementation norm cost, and $m$ for the largest fusion multiplicity. Under the hypotheses of Section~\ref{subsec:localmaps}, given exact charge resolution and recoupling, we can find
\begin{equation}\label{intro:locality}
    q_\ell K_\ell^2\ge\frac{\theta_D}{m+\theta_D},
    \qquad
    \theta_D=\frac{\sqrt{1+D}-1}{\sqrt{1+D}+1},
    \qquad D=D_{\mathrm{assoc}}(\mathcal C).
\end{equation}

For $D>0$, the compression error cannot tend to zero while the implementation cost remains bounded: bounded splitting and merging retractions will give us a stronger estimate. When the charged codes lie in a common isolated degenerate band of the same Hamiltonian, the local operators also yield a perturbation with nonzero first-order splitting.

We will also consider a measurement-based construction using coherent superpositions of fusion channels instead and obtain a second way to test compatibility with positive elementary metrics.

\begin{theorem}\label{intro:operational}
Let $P_\alpha$ be the two placements of the calibrated projections on $H_3=\bigoplus_d\operatorname{Hom}(d,X^{\otimes3})$. For any injective encoding $A:H_3\to\mathcal H$ and physical effects $0\le E_\alpha\le I$, define
\[
    \eta=\max_\alpha
    \|(E_\alpha A-AP_\alpha)(A^\dagger A)^{-1/2}\|.
\]
Then, with $N=\max_{a,b,c,d}\dim\operatorname{Hom}(d,a\otimes b\otimes c)$ and $D=D_{\mathrm{assoc}}(\mathcal C)$,
\begin{equation}\label{intro:operationalbound}
    12N\eta(1+2\eta)^5\ge\frac{D}{D+2}.
\end{equation}

\end{theorem}

The superposition measurements determines the adjoints of the off-diagonal matrix units, which allows us to compare the physical inner products on the two fusion trees, hence being able to relate the measurement error to the associator defect.
In some cases, two selectors already suffice. As an example, for the diagonal doubled Yang--Lee charge, the pair-vacuum selectors satisfy $PQP=\varphi^4P$, whereas orthogonal projections must satisfy $0\le PQP\le P$. This is also expanded upon in Section~\ref{subsec:YLmeasure}.

\subsection{Guide to the paper}
Section~\ref{part:metrics} develops the theory of fusion metrics, including the Yang--Lee calculation and the variational problem. Section~\ref{part:products} studies products and centers, with examples from doubled Yang--Lee and quantum-group categories. Sections~\ref{subsec:chargedrealizations}--\ref{subsec:chargedband} and \ref{part:measurements} discuss the two physical applications, to local splitting operations and fusion measurements, respectively. These sections use different hypotheses and can be read independently. Readers interested in the measurement problem may start with the three-charge example in Section~\ref{subsec:measurementproblem}. The appendices contain technical proofs and calculations.
\subsection{Acknowledgements}
I would like to thank Du Pei for suggesting the problem and his continued support and mentorship for this entire project. I also thank an anonymous colleague for their suggestions, some of which were incorporated here, and others which will be expanded on in a followup.

LLMs (ChatGPT OpenAI) was used solely for LaTeX formatting and typesetting, limited proofreading, and preliminary literature searches.

%% file: 02-fusion-metrics.tex
\section{Fusion metrics and variational theory}\label{part:metrics}

Variables here are inner products on the elementary fusion spaces. Even for Yang--Lee case, these variables can diverge while simultaneously the normalized associator stays fixed.

\subsection{Fusion metrics, unitarity, and the Yang--Lee example}\label{sec:fusion}
Fix a complex fusion category $\mathcal C$, representatives $I$ of its simple objects, and normalized unitors. Set
\[
 V_{ab}^c=\operatorname{Hom}(c,a\otimes b),\qquad
 N_{ab}^c=\dim V_{ab}^c.
\]

An elementary vector $v\in V_{ab}^e$ and a vector $w\in V_{ec}^d$ composes to the morphism
\[
 d\xrightarrow{\ w\ }e\otimes c
 \xrightarrow{\ v\otimes\id_c\ }(a\otimes b)\otimes c.
\]

Semisimplicity expresses every morphism to this parenthesized tensor product as a sum of such composites. The right parenthesization has the analogous decomposition. Consequently the associator is represented by a block
\begin{equation}\label{fusion:trees}
 F_B:\bigoplus_e V_{ab}^e\otimes V_{ec}^d
 \longrightarrow\bigoplus_f V_{bc}^f\otimes V_{af}^d,
 \qquad B=(a,b,c;d).
\end{equation}

Both sides of \eqref{fusion:trees} have dimension
\[
n_B=\sum_e N_{ab}^eN_{ec}^d
    =\sum_f N_{bc}^fN_{af}^d.
\]

Now, choose a positive Hermitian form $g_{ab}^c$ on each nonzero elementary space and fix unit forms to one. Within a fixed channel, the forms at the two vertices will be tensored. Thus,
\begin{equation}\label{fusion:metrics}
    G_{L,B}=\bigoplus_e g_{ab}^e\otimes g_{ec}^d,
    \qquad
    G_{R,B}=\bigoplus_f g_{bc}^f\otimes g_{af}^d.
\end{equation}

Each occurrence of $V_{ab}^c$ carries the same form, regardless of the other labels in the tree. We will call this reuse of elementary forms to be coherence.

The tree metrics $G_{L,B}$ and $G_{R,B}$ cannot therefore be independent variables.

\begin{proposition}\label{fusion:commonmetric}
A complex fusion category is unitarizable if and only if it admits normalized positive elementary forms for which
\begin{equation}\label{fusion:isometry}
    F_B^\dagger G_{R,B}F_B=G_{L,B}
\end{equation}
for every associator block.
\end{proposition}

\begin{proof}
For a unitary category, orthonormal bases in the spaces would induce orthonormal tree bases. The associator is unitary in these bases, and a change to the fixed fusion bases gives \eqref{fusion:isometry}.

Conversely, take $S_{ab}^c=(g_{ab}^c)^{1/2}$. As positive square roots commute with tensor products and orthogonal direct sums, so the tree changes are $S_L=G_L^{1/2}$ and $S_R=G_R^{1/2}$, and \[U_B=S_RF_BS_L^{-1},\qquad U_B^\dagger U_B=I.\]

These changes are made at each fusion vertex, so the pentagon and unit equations are still true.

Now, give morphisms the ordinary Hilbert adjoint on their multiplicity spaces.

In the chosen fusion bases, tensor products of morphisms are from direct sums and Kronecker products, so
\[(f\otimes g)^\dagger=f^\dagger\otimes g^\dagger.\]
The unitary fusion category is monoidally equivalent to $\mathcal C$ \cite[Remark 3.14(iii)]{FG}.
\end{proof}

As for a coherent positive family, set
\begin{equation}\label{direct:eq:defect}
    U_B(g)=G_{R,B}^{1/2}F_BG_{L,B}^{-1/2},\qquad
    \delta(g)=\max_B\|U_B(g)^\dagger U_B(g)-I\|_{\mathrm{op}},
    \qquad \Def(\mathcal C)=\inf_{g>0}\delta(g).
\end{equation}
When $s_{B,j}$ are the singular values of $U_B(g)$, the block error is $\max_j|s_{B,j}^2-1|$.

The normalization has also a useful expression in the original tree coordinates:
\[
    \|U_B^\dagger U_B-I\|
    =\sup_{v\ne0}
    \frac{|(F_Bv)^\dagger G_{R,B}(F_Bv)-v^\dagger G_{L,B}v|}
         {v^\dagger G_{L,B}v}.
\]
The choice of fusion bases does not affect $\Def$. By elementary polar decomposition, optimization over positive forms is equivalent to singular-value optimization over the normalized complex gauge orbit. The coordinate calculation is explored in Appendix~\ref{app:gauges}

\begin{lemma}\label{fusion:rigidlimit}
Let $F^{(k)}$ be coherent normalized gauges of a fixed fusion category. Assume that every block converges and that the inverse block norms are uniformly bounded: within the same elementary gauge orbit, the limiting blocks define a rigid tensor structure.
\end{lemma}

\begin{proof}
Let $F^{(\infty)}$ be the blockwise limit. The uniform inverse bounds implies that each limiting block is invertible and the inverses converge.

For a simple $i$, $\bar i$ denotes its dual label, write
\[
    a_i(F)=[F_{i\bar i i}^{i}]_{1,1}.
\]
The first triangle identity says that $a_i(F)$ is nonzero. In a normalized elementary gauge,
\begin{equation}\label{fusion:vacuumgauge}
    a_i(S\cdot F)=
    \frac{s_{\bar i i}^{1}}{s_{i\bar i}^{1}}a_i(F).
\end{equation}
The product $a_i(F^{(k)})a_{\bar i}(F^{(k)})$ will be gauge invariant and nonzero. Since both factors are entries of convergent blocks, neither can vanish at the limit.

Fix the vacuum inclusion $j_i:1\to i\otimes\bar i$ and projection $p_i:\bar i\otimes i\to1$. At the stage $k$, take the coevaluation to be $j_i$ and the evaluation to be $a_i(F^{(k)})^{-1}p_i$, so the first triangle identity would obviously hold by the definition of $a_i$.

For the other triangle, use the rigidity argument for each finite stage. Any dual pair has a coevaluation proportional to $j_i$ as the vacuum fusion space is one dimensional, so by scaling the pair such that the coevaluation equals $j_i$, the first triangle fixes its evaluation uniquely and thus must be the one chosen. Since $a_i(F^{(k)})$ has a nonzero limit, the duality maps converge, and thus both triangle identities hold at the limit. Then, direct sums will give the duals for all objects.

By Ocneanu rigidity, the gauge orbit of \(F^{(\infty)}\) is open in the rigid locus \cite[Theorem 9.1.4]{EGNO}, hence $F^{(k)}$ belongs to this orbit for all sufficiently large $k$. And since every $F^{(k)}$ also belongs to the original orbit, the two orbits coincide.
\end{proof}

\begin{theorem}\label{fusion:zero}
For every complex fusion category,
\begin{equation}\label{fusion:zeroeq}
    \Def(\mathcal C)=0
    \quad\Longleftrightarrow\quad
    \mathcal C\text{ is unitarizable}.
\end{equation}
\end{theorem}

\begin{proof}
An exactly compatible metric has zero error by Proposition~\ref{fusion:commonmetric}.

For the converse, choose forms $g^{(k)}$ with $\delta(g^{(k)})\le\varepsilon_k\to0$. Their normalized associators obey
\[
    (1-\varepsilon_k)I\le
    U_B(g^{(k)})^\dagger U_B(g^{(k)})
    \le(1+\varepsilon_k)I.
\]

After discarding finitely many terms, the singular values are bounded above and away from zero uniformly in every block. As there are finitely many entries, a common subsequence converges to a unitary family in the original gauge orbit by Lemma~\ref{fusion:rigidlimit}.

The polar decomposition Lemma~\ref{fusion:rigidlimit} would give positive elementary forms making every block unitary, where Proposition~\ref{fusion:commonmetric} applies.
\end{proof}

\begin{corollary}\label{fusion:ringgap}
For a fixed based fusion ring $N$, there exists $\delta_N>0$ such that every complex fusion category with fusion rules $N$ is either unitarizable or must have an associator defect at least $\delta_N$.
\end{corollary}

\begin{proof}
There are finitely many tensor-equivalence classes with these fusion rules \cite[Theorem 9.1.4]{EGNO}. If any are nonunitarizable, take the minimum of their positive defects, using Theorem~\ref{fusion:zero}; if none are, any positive $\delta_N$ satisfies the assertion
\end{proof}

Now, let us take a look at the Yang--Lee case. The Yang--Lee fusion rule is $\tau\otimes\tau=1\oplus\tau$.  We use the symmetric gauge:

\begin{equation}\label{square:eq:YL}
    F_{\tau}^{\tau\tau\tau}
    =\begin{pmatrix}-\phi&i\sqrt\phi\\i\sqrt\phi&\phi\end{pmatrix},
    \qquad \phi=\frac{1+\sqrt5}{2},
\end{equation}
with every other admissible block to be scalar one, as given by \cite[Eqs. (6)--(7)]{Lootens}. There are only two nonunit elementary forms, namely $u=g_{\tau\tau}^{1}>0$ and $v=g_{\tau\tau}^{\tau}>0$. As for the nontrivial block both tree metrics are $\operatorname{diag}(u,v^2)$, so their ratio $r=v^2/u$ gives
  \begin{equation}\label{fusion:YLr}
    U(r)=\begin{pmatrix}
    -\phi&i\sqrt\phi\,r^{-1/2}\\
    i\sqrt\phi\,r^{1/2}&\phi
    \end{pmatrix}.
\end{equation}

All scalar blocks remain isometric for every \(u,v>0\), and thus the positive-metric problem will reduce to the single parameter \(r\).

\begin{theorem}\label{fusion:YLminimum}
Writing $M=\phi+\sqrt\phi$, one obtains
\begin{equation}\label{fusion:YLvalue}
    \Def(\mathrm{YL})=M^2-1.
\end{equation}
Hence, the minimizing elementary metrics are exactly those with $u=v^2$.
\end{theorem}

\begin{proof}
Since $\phi^2-\phi=1$, the determinant of $U(r)$ is $-1$.
The squared singular values are therefore $\mu(r)$ and $\mu(r)^{-1}$,
with $\mu(r)>1$. Their sum is
\begin{equation}\label{fusion:YLtrace}
    \mu(r)+\mu(r)^{-1}
    =\operatorname{Tr}(U(r)^\dagger U(r))
    =2\phi^2+\phi(r+r^{-1}).
\end{equation}
Since $\mu+\mu^{-1}$ is strictly increasing for $\mu>1$, and $r+r^{-1}\geq2$ with equality only at $r=1$, the larger squared singular value is minimized uniquely at $r=1$, so the defect is $\mu(r)-1$, since $\mu-1\geq1-\mu^{-1}$.

As for $r=1$, the singular values are $\phi+\sqrt\phi$ and $\phi-\sqrt\phi$ with a product of one. The remaining scalar blocks are already isometric, thus proving \eqref{fusion:YLvalue}.
\end{proof}

The condition $u=v^2$ fixes only the ratio of the elementary metrics. By taking $v\to\infty$ with $u=v^2$, the defect is at its minimum, although the metrics themselves would diverge.

An interesting result is also that the stability estimate is \[|\log r|\leq \sqrt{\frac{\delta(r)-\delta_*}{\phi}}, \qquad \delta_*=M^2-1. \] and thus any sequence of nearly minimizing metrics must have $r\to1$.

\begin{proposition}\label{fusion:restriction}
If $\mathcal A\subset\mathcal C$ is a full fusion subcategory, then
\begin{equation}\label{fusion:restrictioneq}
    \Def(\mathcal C)\ge\Def(\mathcal A).
\end{equation}
If $\mathcal U$ is unitarizable, then
\begin{equation}\label{fusion:unitfactor}
    \Def(\mathcal C\boxtimes\mathcal U)=\Def(\mathcal C).
\end{equation}
\end{proposition}

\begin{proof}
Restrict elementary metrics from $\mathcal C$ to $\mathcal A$. The maximum defect cannot increase, so taking infima gives \eqref{fusion:restrictioneq}. Restriction to $\mathcal C\boxtimes\mathbf 1$ gives the lower bound in \eqref{fusion:unitfactor}.

For the upper bound, now choose exactly unitary elementary metrics on $\mathcal U$ and an arbitrary coherent family on $\mathcal C$ such that their tensor products are now coherent positive metrics on the Deligne product.

Up to unitary permutations of the basis, each normalized block is \(U_B\otimes V_E\) for a unitary $V_E$, hence
\[
    (U_B\otimes V_E)^\dagger(U_B\otimes V_E)-I
    =(U_B^\dagger U_B-I)\otimes I,
\]
with a norm of $\|U_B^\dagger U_B-I\|$. The maximum product error is thus the original error on $\mathcal C$. Minimizing these forms gives the upper bound in \eqref{fusion:unitfactor}.

\end{proof}

\subsection{Determinant balance and variational geometry}\label{sec:direct}
\subsubsection{The weighted pentagon identity}
A bounded Gram defect gives upper bounds on the singular values of the normalized associators, but no immediate lower bounds is present. The pentagon identity therefore provides the missing control through a determinant balance weighted by the Frobenius--Perron dimensions.

Write
\[
    d_a=\operatorname{FPdim}(a),\quad
    \Phi=\operatorname{FPdim}(\mathcal C)=\sum_a d_a^2,\quad
    w_B=d_ad_bd_cd_d,\quad W=\sum_Bw_Bn_B.
\]

\begin{figure}[tbp]
\centering
\begin{tikzpicture}[x=1cm,y=1cm,>=Stealth,font=\small,
  every node/.style={inner sep=4pt},line width=.75pt]
\node (A) at (-4,1.6) {$((ab)c)d$};
\node (B) at (0,2.3) {$(a(bc))d$};
\node (C) at (4,1.6) {$a((bc)d)$};
\node (D) at (2.5,-.7) {$a(b(cd))$};
\node (E) at (-2.5,-.7) {$(ab)(cd)$};
\draw[->] (A)--node[above left] {$F\otimes I$} (B);
\draw[->] (B)--node[above right] {$F$} (C);
\draw[->] (C)--node[right] {$I\otimes F$} (D);
\draw[->] (A)--node[left] {$F$} (E);
\draw[->] (E)--node[below] {$F$} (D);
\end{tikzpicture}
\caption{The pentagon identity is here used to establish determinant balance}
\label{fig:pentagon-balance}
\end{figure}

\begin{proposition}\label{direct:statement:1}
For every exact normalized associator family,
\begin{equation}\label{direct:eq:determinants}
    \sum_{a,b,c,d}d_ad_bd_cd_d\log|\det F_{abc}^{d}|=0,
    \qquad W=\Phi^3.
\end{equation}
In particular, the identity holds in every coherent elementary gauge including that of the matrices $U_B(g)$.
\end{proposition}

\begin{proof}
Let $\ell(a,b,c;d)=\log|\det F_{abc}^{d}|$. By taking absolute
determinants of the four-object pentagon, with inputs $a,b,c,d$ and output $h$, one obtains
\begin{equation}\label{direct:eq:detpentagon}
\begin{aligned}
&\sum_x N_{xd}^{h}\ell(a,b,c;x)
+\sum_y N_{bc}^{y}\ell(a,y,d;h)
+\sum_z N_{az}^{h}\ell(b,c,d;z)\\
&\qquad =
\sum_u N_{ab}^{u}\ell(u,c,d;h)
+\sum_v N_{cd}^{v}\ell(a,b,v;h).
\end{aligned}
\end{equation}
The multiplicities are the dimensions of the spectator spaces: for example, the first edge is a direct sum of $F_{abc}^{x}\otimes I_{V_{xd}^{h}}$. This formula also would include arbitrary fusion multiplicities as basis permutations would have an absolute determinant of one.

Let $L=\sum_Bw_B\ell(B)$. By multiplying by $d_ad_bd_cd_dd_h$ and summing over the external labels, the Frobenius--Perron identities  \cite[Section 3.3]{EGNO}
\begin{equation}\label{direct:eq:FPidentities}
    \sum_hN_{xd}^{h}d_h=d_xd_d,\qquad
    \sum_{b,c}d_bd_cN_{bc}^{y}=\Phi d_y
\end{equation}
make each edge contribute $\Phi L$. The second identity would follow from reciprocity and the first. Thus, the pentagon would give $3\Phi L=2\Phi L$, hence $L=0$. Finally, by summing the two vertex dimensions, one obtains $W=\sum_e(\Phi d_e)^2=\Phi^3$.
\end{proof}

\begin{corollary}\label{direct:statement:2}
If $\|F_B\|_{\mathrm{op}}\leq M$ for every block, then $M\geq1$ and
\begin{equation}\label{direct:eq:inversebound}
    s_{\min}(F_B)\geq M^{-(W-w_B)/w_B},\qquad
    \|F_B^{-1}\|_{\mathrm{op}}\leq M^{W/w_B-1}.
\end{equation}
For $M=1$, every block is unitary.
\end{corollary}

\begin{proof}
List all the singular values $s_{B,j}$ and assign each the positive weight $w_B$. \eqref{direct:eq:determinants} becomes
\[
    \sum_{B,j}w_B\log s_{B,j}=0.
\]

If $M<1$, the sum would be negative. For a smallest singular value $s$ in block $B$, every other term is at most its weight times $\log M$. Therefore
\[
    0\leq w_B\log s+(W-w_B)\log M.
\]
which proves \eqref{direct:eq:inversebound}.

If $M=1$, all log singular values are nonpositive and their positive weighted sum will be zero.

\end{proof}

\subsubsection{Attainment and the role of the gauge orbit}
\begin{theorem}\label{direct:statement:3}

For every complex fusion category there are finite positive elementary metrics $g_*$ such that
\begin{equation}\label{direct:eq:attained}
    D_{\mathrm{assoc}}(\mathcal C)=\delta(g_*).
\end{equation}
More generally, the normalized coherent gauge orbit of $F$ is closed in the finite-dimensional vector space of all its associator entries where every bounded closed subset of this orbit is compact.
\end{theorem}

\begin{proof}
First, we will prove that a coherent gauge orbit is closed.

Suppose $F^{(k)}=S^{(k)}\cdot F$ converges. The forward block norms are uniformly bounded, so Corollary~\ref{direct:statement:2} bounds their inverses. Thus, by Lemma~\ref{fusion:rigidlimit}, the limit lies in the same gauge orbit.

Now,  let $g^{(k)}$ be a minimizing sequence with $\delta(g^{(k)})\le C$ where the normalized blocks satisfy
\[
    \|U_B(g^{(k)})\|\le\sqrt{1+C}.
\]

Since there are finitely many blocks, a subsequence converges, so that the closedness of gauge orbit would give a limiting coherent gauge $S\cdot F$. Continuity shows it attains the infimum.

Finally, write $S=VP$ by elementary polar decomposition. The induced tree factors give
\[
    (S\cdot F)_B
    =V_{R,B}(P\cdot F)_BV_{L,B}^\dagger.
\]

As the unitary factors do not change the defect, so the positive elementary metrics $g=P^2$ attain the minimum.

\end{proof}

\begin{corollary}[]\label{direct:statement:4}
The following conditions are equivalent:
\begin{enumerate}
\item $\mathcal C$ is unitarizable;
\item $D_{\mathrm{assoc}}(\mathcal C)=0$;
\item there are coherent positive elementary metrics such that
      $F_B^\dagger G_{R,B}F_B\leq G_{L,B}$ for every block;
\item $\inf_g\max_B\|U_B(g)\|_{\mathrm{op}}=1$.
\end{enumerate}
\end{corollary}

Indeed condition 3 says that every $U_B(g)$ is a contraction. Since the  Frobenius--Perron weights are positive, determinant balance would force $|\det U_B(g)|=1$ for every block, hence every $U_B(g)$ is unitary. By the same compactness argument, one can obtain the forward-norm infimum. The standard equivalence between coherent unitary $F$ matrices and a unitary fusion structure then applies; see \cite[Remark 3.14(iii)]{FG}.

\subsubsection{Geodesics and active constraints}
We will use product matrix geodesics
\[
    (g\mathbin{\#_t}h)_v
    =g_v^{1/2}(g_v^{-1/2}h_vg_v^{-1/2})^t g_v^{1/2},\qquad 0\leq t\leq1.
\]

Choose elementary bases orthonormal for $g$ and $h_v=e^{2H_v}$ with $H_v=H_v^\dagger$. The normalized blocks has the coherent gauge form

\begin{equation}\label{direct:eq:ray}
    F_B(t)=e^{tH_{R,B}}F_Be^{-tH_{L,B}},
    \quad H_v=H_v^\dagger,
\end{equation}
where
\[
    H_{L,B}=\bigoplus_e(H_{ab}^{e}\otimes I+I\otimes H_{ec}^{d}),
\]
and similarly for $H_{R,B}$. Unit generators vanish.

\begin{lemma}\label{direct:statement:5}
    Along \eqref{direct:eq:ray}, both
\[
    \log\|F_B(t)\|_{\mathrm{op}},\qquad
    \log\|F_B(t)^{-1}\|_{\mathrm{op}}
\]
are convex functions of $t$.
\end{lemma}

\begin{proof}
    For Hermitian $A,B$, consider $\langle u,e^{zA}Xe^{-zB}v\rangle$ on the strip $0\leq\operatorname{Re}z\leq1$.
    The imaginary exponentials are unitary, so the boundary values are bounded by $\|X\|$ and $\|e^AXe^{-B}\|$, respectively; the function is also bounded throughout the strip.

    The scalar three-lines theorem \cite[Theorem 7.13]{Tropp} gives, when we take the supremum over unit $u,v$,
    \[
        \|e^{tA}Xe^{-tB}\|
        \leq \|X\|^{1-t}\|e^AXe^{-B}\|^t.
    \]
    The same argument applies on any subinterval--- also, to $F_B(t)^{-1}=e^{tH_{L,B}}F_B^{-1}e^{-tH_{R,B}}$.
\end{proof}

Throughout this subsection, we will denote the operator norm with \(\|\cdot\|\) for convenience.
For $c\geq0$,
\begin{equation}\label{direct:eq:levelobjective}
    \Psi_c(g)=
    \begin{cases}
    \displaystyle\max_B\left\{
    \log\|U_B(g)\|-\tfrac12\log(1+c),\,
    \log\|U_B(g)^{-1}\|+\tfrac12\log(1-c)\right\},&c<1,\\[2mm]
    \displaystyle\max_B\left\{
    \log\|U_B(g)\|-\tfrac12\log(1+c)\right\},&c\geq1.
    \end{cases}
\end{equation}

All norms here and below must be operator norms. Then
\begin{equation}\label{direct:eq:sublevel}
    \delta(g)\leq c\ \Longleftrightarrow\ \Psi_c(g)\leq0,
    \qquad
    \delta(g)<c\ \Longleftrightarrow\ \Psi_c(g)<0\quad(c>0).
\end{equation}
For $c\geq1$, the lower singular-value constraint is vacuous. Lemma~\ref{direct:statement:5} then shows that $\Psi_c$ is geodesically convex, and hence are sublevel sets of the defect.

\begin{theorem}
Every local minimum of $\delta$ is a global minimum. Fix a metric \(g\) with \(c=\delta(g)>0\). Choose elementary bases orthonormal with respect to \(g\).

For a block that satisfies \(\|F_B\|^2=1+c\), let \(F_Bv=\|F_B\|u\) be any unit top singular pair, and let
\begin{equation}\label{direct:eq:forwardactive}
    L_{B,u,v}^{+}(H) =\langle u,H_{R,B}u\rangle-\langle v,H_{L,B}v\rangle.
\end{equation}

If \(c<1\), the blocks with \(\|F_B^{-1}\|^2=(1-c)^{-1}\) give analogous functionals. For every unit top singular pair $F_B^{-1}v=\|F_B^{-1}\|u$, let
\begin{equation}\label{direct:eq:inverseactive}
L_{B,u,v}^{-}(H)
=\langle u,H_{L,B}u\rangle-\langle v,H_{R,B}v\rangle.
\end{equation}

Let $\mathcal A$ be the set the functionals. Then $g$ is globally optimal if and only if
\begin{equation}\label{direct:eq:certificate}
0\in\operatorname{conv}(\mathcal A).
\end{equation}

Equivalently, there are $L_i\in\mathcal A$ and $\lambda_i\geq0$ such that $\sum_i\lambda_i=1$ and $\sum_i\lambda_iL_i=0$.

\end{theorem}

\begin{proof}
The variational formula applied to $e^{tH_R}F_Be^{-tH_L}$ and its inverse gives
\begin{equation}\label{direct:eq:directional}
\left.\frac{d}{dt^+}\Psi_c(e^{tH}\cdot F)\right|_{t=0}
=\max_{L\in\mathcal A}L(H).
\end{equation}
If \(0\in\operatorname{conv}(\mathcal A)\), the derivative is nonnegative in every direction. Geodesic convexity gives \(\Psi_c\geq0\).

Conversely, it is first apparent that set \(\mathcal A\) is compact. If $0\notin\operatorname{conv}(\mathcal A)$, the strict separation gives a direction $H$ such that \[L(H)<0\qquad\text{for every }L\in\mathcal A. \]
Thus \(\Psi_c(e^{tH}\cdot F)<0\) for sufficiently small \(t>0\), contradicting local minimality. Thus every local minimum is global.
\end{proof}

\subsubsection{When forward norms suffice}
Define \[A(\mathcal C)=\min_g\bigl(\max_B\|U_B(g)\|^2-1\bigr), \qquad K=\max_B(\Phi^3/w_B-1). \]
The minimum is attained and it is clear that  $A\geq0$.

By Corollary~\ref{direct:statement:2},
\begin{equation}\label{direct:eq:comparison}
A\leq D_{\mathrm{assoc}}\leq
\max\{A,\,1-(1+A)^{-K}\}.
\end{equation}

It follows that
\begin{equation}\label{direct:eq:largeD}
D_{\mathrm{assoc}}\geq1\ \Longleftrightarrow\ A\geq1, \qquad D_{\mathrm{assoc}}=A\ \text{in that case}.
\end{equation}
This is not true always for $A<1$, however.

%% file: 03-products-centers.tex
\section{Products, centers, and nonunitary families}\label{part:products}

\subsection{Coupled metrics and product optimality}\label{sec:coupled}
Define the attained forward-norm invariant
\[
    a(\mathcal C)=\min_g\max_B\log\|U_B(g)\|.
\]
It is nonnegative by determinant balance. Choose forward-norm minimizers on $\mathcal C$ and $\mathcal D$, and let $g_0$ be their product metric. Since the normalized blocks at $g_0$ are tensor products, up to channel permutations,
\[
    \max_{B,C}\log\|U_{BC}(g_0)\|
    =a(\mathcal C)+a(\mathcal D)=:a_0.
\]
Product metrics need not exhaust the elementary metrics on $\mathcal C\boxtimes\mathcal D$.

Let $\mathscr H$ be the real space of Hermitian elementary gauge generators on $\mathcal C\boxtimes\mathcal D$ that vanishes on unit vertices, with inner product $\sum_p\Tr(H_pK_p)$.

For an active block $F$ with $\|F\|=e^{a_0}$, choose a unit right singular vector $v$ corresponding to $\|F\|$, and set $u=Fv/e^{a_0}$. Define its score $S(F,v)\in\mathscr H$:
\begin{equation}\label{coupled:score}
    \langle S(F,v),H\rangle
    =\langle u,H_Ru\rangle-\langle v,H_Lv\rangle.
\end{equation}

Let $\mathscr S$ be the convex hull of all such scores and set
\[
    \mu=\min_{S\in\mathscr S}\|S\|_{\rm HS}.
\]
The distinction between projected and full balance is illustrated below Figure~\ref{fig:coupled-scores}.

\begin{figure}[tbp]
\centering
\begin{tikzpicture}[x=1cm,y=1cm,>=Stealth,font=\small,line width=.7pt]
\begin{scope}[shift={(-3.3,0)}]
\draw[->,gray] (-2.1,0)--(2.1,0) node[below] {$\mathscr P_0$};
\draw[->,gray] (0,-.5)--(0,2.9);
\filldraw[fill=diagramblue!12,draw=diagramblue] (-1.5,1.3)--(1.5,1.3)--(.6,2.6)--cycle;
\draw[densely dashed] (-1.5,1.3)--(-1.5,0);
\draw[densely dashed] (1.5,1.3)--(1.5,0);
\draw[line width=2pt] (-1.5,0)--(1.5,0);
\fill (0,0) circle(1.8pt) node[below left] {$0$};
\fill[diagramblue] (0,1.3) circle(2pt) node[right] {$S_*$};
\draw[->,diagramblue,line width=1pt] (0,1.25)--(0,.15);
\node at (.3,2) {$\mathscr S$};
\node[align=center] at (0,-1.05) {factorwise balance\\$0\in\operatorname{proj}_{\mathscr P_0}\mathscr S$};
\end{scope}
\begin{scope}[shift={(3.3,0)}]
\draw[->,gray] (-2.1,0)--(2.1,0);
\draw[->,gray] (0,-.5)--(0,2.9);
\filldraw[fill=diagramblue!12,draw=diagramblue] (-1.6,-.35)--(1.5,-.35)--(.35,2.5)--cycle;
\fill (0,0) circle(1.8pt) node[above right] {$0$};
\node at (.25,1.3) {$\mathscr S$};
\node[align=center] at (0,-1.05) {full balance\\$0\in\mathscr S$};
\end{scope}
\end{tikzpicture}
\caption{Factorwise balance (left) and full balance (right).}\label{fig:coupled-scores}
\end{figure}

\begin{theorem}\label{coupled:criterion}
The product metric $g_0$ minimizes the forward norm over all coherent metrics if and only if $\mu=0$.
If $\mu>0$, set $S_*$ be the closest point of $\mathscr S$ to be zero and let $H=-S_*/\mu$. Then,
\begin{equation}\label{coupled:margin}
    \left.\frac{d}{dt^+}
    \max_{B,C}\log\|e^{tH_R}F_{BC}e^{-tH_L}\|
    \right|_{t=0}=-\mu<0.
\end{equation}
If $c_0=e^{2a_0}-1\geq1$, then $g_0$ would also minimizes the associator defect if and only if $\mu=0$. $\mu>0$ implies
\[
    D_{\mathrm{assoc}}(\mathcal C\boxtimes\mathcal D)<c_0.
\]
\end{theorem}
\begin{proof}
By Section~\ref{sec:direct}, the directional derivative along $H$ is
\[
    \max_{S\in\mathscr S}\langle S,H\rangle.
\]
If $0\in\mathscr S$, for every $H$ this is nonnegative. Global minimality follows from the geodesic convexity.

Elsewise, set $S_*$ to be the point of $\mathscr S$ closest to the origin given that $\|S_*\|=\mu>0$.
Then \[ \langle S-S_*,S_*\rangle\geq0 \qquad(S\in\mathscr S), \] so choosing $H=-S_*/\mu$ would give \[ \max_{S\in\mathscr S}\langle S,H\rangle=-\mu. \] The forward norm must strictly decrease this direction.

Since each normalized Gram matrix is positive definite, their eigenvalues are strictly positive. Thus, when $c_0\geq1$, the defect at $g_0$ will be affected as a sufficiently small strict decrease of the maximal forward norm must also decrease the defect. Conversely, any metric with defect less than $c_0$ must have smaller maximal forward norm than $g_0$. \end{proof} The factorwise tangent space is \[ \mathscr P_0=\{H_{p,q}=A_p\otimes I+I\otimes B_q\}, \].

Factorwise optimality would give $0\in\operatorname{proj}_{\mathscr P_0}\mathscr S$, while full optimality requires $0\in\mathscr S$.

If there is a single active block that has a simple top singular value, its score $S$ will be orthogonal to $\mathscr P_0$; when $S\neq0$, the direction $-S$ will give strict coupled descent instead.

The condition $0\in\mathscr S$ is a semidefinite feasibility problem. For each active block, one can choose $Q_B\geq0$ supported on its top right singular space with $\sum_B\Tr Q_B=1$, and require the total score obtained from $Q_B$ and $F_BQ_BF_B^\dagger/e^{2a_0}$ to vanish.

It is worth noting that a decreasing direction need not remain decreasing for arbitrarily large steps. After all, active singular value may stop decreasing, or an inactive block may become active; an explicit step bound is given in Proposition~\ref{coupled:gain}.

For two Yang--Lee factors, the factorwise optimum has a single active block with a simple top singular value. Let $a=(\tau,\tau)$, $p=(\one,\tau)$, and $r=(\tau,\one)$. Consider
\[
    H_{pa}^{a}=H_{ra}^{a}=1,\qquad H_{ap}^{a}=H_{ar}^{a}=-1,
\]
with other entries zero. $H_L=\operatorname{diag}(0,1,1,0)$ and $H_R=-H_L$ on the active block. The top singular vectors have uniform channel probabilities, so the score in \eqref{coupled:score} is $-1$, giving the strict descent by Theorem~\ref{coupled:criterion}. At the coupled minimum, two additional blocks will become active (Section~\ref{sec:square}); as for defects less than one, one should also consider that inverse-norm constraints may also be active and must be included in the balance.

\subsection{The doubled Yang--Lee category and tensor powers}\label{sec:square}
Consider the product of two Yang--Lee categories---all fusion multiplicities are one---where the elementary metric is a positive scalar. Let $S_v=\sqrt{g_v}>0$. The action on an associator entry is
\[
    F_{abc}^d(S)_{fe}=\frac{S_{bc}^fS_{af}^d}{S_{ab}^eS_{ec}^d}(F_{abc}^d)_{fe}.
\]
Scalar phases do not affect the singular values (see Appendix~\ref{app:gauges}), so these variables cover all possible metrics, including those coupling the two factors. We shall use \eqref{square:eq:YL} for the nonscalar one-factor block and scalar one for the others.

Let
\[
    a=(\tau,\tau),\qquad p=(1,\tau),\qquad r=(\tau,1).
\]
The fusion rules are
\[
    \begin{array}{c|ccc}
    \otimes&p&r&a\\\hline
    p&1+p&a&r+a\\
    r&a&1+r&p+a\\
    a&r+a&p+a&1+p+r+a
    \end{array}
\]
As there are respectively $5,5,8$ elementary channels with nonunit inputs, there is a total of 18 positive metric parameters.

Of the 144 nonempty associator blocks, 121 are scalar, 22 are two-dimensional, and only one is four-dimensional. At the minimum, only three will attain the maximal forward norm. We will show their balance below; see Appendix~\ref{app:double} for the checks of the remaining 141 blocks and all 18 metric directions.

\begin{theorem}\label{square:statement:1}
Define
\begin{equation}\label{square:eq:lambda}
    \lambda=\frac{3\phi+1+\sqrt{(3\phi+1)^2-4}}2,
    \qquad z=\frac{\phi\lambda-1}{\phi^3\lambda-1},\qquad k=\sqrt z.
\end{equation}
Then
\begin{equation}\label{square:eq:exact}
    D_{\mathrm{assoc}}(\mathrm{YL}\boxtimes\mathrm{YL})
    =\lambda^2-1.
\end{equation}
The attaining coherent metric is
\begin{equation}\label{square:eq:metrics}
    g_{pa}^{a}=g_{ra}^{a}=z^{-1},\qquad
    g_{ap}^{a}=g_{ar}^{a}=z,
    \qquad g_v=1\ \text{on every other elementary fusion space}.
\end{equation}
\end{theorem}

The three active norms meet at $\lambda$, as shown in Figure~\ref{fig:doubled-balance}. To establish global optimality, we use positive weights for which the weighted logarithmic derivative vanishes in every elementary metric direction; see \eqref{square:eq:functionals}--\eqref{square:eq:balance}. Geodesic convexity then shows that at least one of these blocks has norm at least $\lambda$ for every metric. The remaining blocks are checked in Appendix~\ref{app:double}.

\begin{figure}[tbp]
\centering
\begin{tikzpicture}[x=1cm,y=1cm,>=Stealth,font=\small,line width=.8pt]
\draw[->] (0,0)--(9.4,0) node[right] {$k$};
\draw[->] (0,0)--(0,4) node[above] {block norm};
\draw[densely dashed,gray] (0,2)--(4,2)--(4,0);
\draw[densely dashed,gray] (8,0)--(8,3.55);
\draw[diagramblue,line width=1.1pt] (.7,1.15) .. controls (2,1.25) and (3,1.55) .. (4,2)
  .. controls (5.6,2.7) and (6.5,3.3) .. (8,3.55);
\draw[line width=1.1pt] (1.3,3.7) .. controls (2.1,2.9) and (3.1,2.25) .. (4,2)
  .. controls (5.4,1.6) and (6.8,1.2) .. (8,1.05);
\fill[diagramblue] (4,2) circle(2.5pt);
\node[left] at (0,2) {$\lambda$};
\node[below] at (4,0) {$k_*$};
\node[below] at (8,0) {$1$};
\node[right,diagramblue] at (8,3.55) {$\|F_0(k)\|$};
\node[right] at (8,1.05) {$\|F_p(k)\|=\|F_r(k)\|$};
\draw[->] (7.8,.5)--(4.3,.5);
\node[above] at (6,.5) {coupled deformation};
\end{tikzpicture}
\caption{Norm balance in the doubled Yang--Lee example.}
\label{fig:doubled-balance}
\end{figure}

\subsubsection{Tensor powers}

What happens if we keep adding Yang--Lee factors?

We can group the copies in pairs and use the doubled metric to find the upper bound. For a lower bound, we look only at the vacuum and nontrivial intermediate channels. As their left and right tree forms agree---even when the metric couples all the copies---its action is a diagonal similarity.

For two arbitrary complex fusion categories $\mathcal C$ and $\mathcal D$, product metrics give
\begin{equation}\label{square:eq:submultiplicative}
    1+D_{\mathrm{assoc}}(\mathcal C\boxtimes\mathcal D)\leq(1+D_{\mathrm{assoc}}(\mathcal C))(1+D_{\mathrm{assoc}}(\mathcal D)).
\end{equation}
Indeed, product elementary metrics give normalized associator blocks $U\otimes V$, up to unitary permutations. Let $X=U^\dagger U$ and $Y=V^\dagger V$, such that
\[
    \|X\otimes Y-I\|\leq\|X-I\|\|Y\|+\|Y-I\|\leq x(1+y)+y,
\]
for defects $x,y$.

Let $D_n=D_{\mathrm{assoc}}(\mathrm{YL}^{\boxtimes n})$. For every integer $n\geq1$,
\begin{equation}\label{square:eq:powers}
    (\phi^n+\phi^{n/2})^2\leq1+D_n\leq\lambda^{2\lfloor n/2\rfloor}M^{2(n\bmod2)}.
\end{equation}

Starting with the lower bound, let $a=(\tau,\ldots,\tau)$ and consider the block $(a,a,a;a)$, restricted to the intermediate channels $1$ and $a$ on both sides such that the left and right tree metrics coincide on these channels, with diagonal entries $g_{aa}^{1}$ and $(g_{aa}^{a})^2$. The corner is diagonally similar to
\[
    T_n=\begin{pmatrix}(-\phi)^n&i^n\phi^{n/2}\\ i^n\phi^{n/2}&\phi^n\end{pmatrix}.
\]
Set $A=\phi^n$ and $B=\phi^{n/2}$. Under positive diagonal similarity, the off-diagonal entries become $i^nBr$ and $i^nB/r$ for some $r>0$. The squared Hilbert--Schmidt norm satisfies
\[
    2A^2+B^2(r^2+r^{-2})\geq2(A^2+B^2),
\]
while the absolute determinant remains $A^2-B^2$. For a $2\times2$ matrix with fixed absolute determinant, the largest squared singular value will increase with the squared Hilbert--Schmidt norm. The minimum is therefore attained at $r=1$, where the largest singular value is $A+B$. And since compression cannot increase operator norm, $D_n\geq(A+B)^2-1$.

Now for the upper bound, group the $n$ factors into pairs and apply Theorem~\ref{square:statement:1}, together with the one-factor optimum when $n$ is odd, and \eqref{square:eq:submultiplicative}.

$\log(1+D_n)$ is subadditive by \eqref{square:eq:submultiplicative}, so Fekete's lemma gives
\begin{equation}\label{square:eq:rate}
    \rho=\lim_{n\to\infty}(1+D_n)^{1/n}=\inf_{n\geq1}(1+D_n)^{1/n}.
\end{equation}

The same argument will give the existence of $\lim_n(1+D_{\mathrm{assoc}}(\mathcal C^{\boxtimes n}))^{1/n}$ for any complex fusion category $\mathcal C$.

\subsubsection{Permutation symmetry and equivariantization}\label{sec:symmetry}

Suppose now that the metric problem has a finite symmetry group. Can we choose a minimizing metric with that same symmetry? Ordinary averaging is a pretty awkward here since tree metrics contain products of elementary forms which obviously would be problematic. So, instead, we use the geometric barycenter. The defect sublevels are geodesically convex, so averaging in this sense keeps us in the same sublevel. As for permutations of Deligne factors, we can then use the invariant forms on equivariant objects.

Let $\mathscr P$ be the product of positive elementary metric cones---omit normalized unit factors--- and let $\delta(g)$ denote the defect at $g$. Equip each factor with the trace-metric distance
\[
    d(g,h)=\|\log(g^{-1/2}hg^{-1/2})\|_{\mathrm{HS}}.
\]
With the product metric, $\mathscr P$ is a complete CAT(0) space. Therefore, every finite set has a unique least-squares barycenter in its closed geodesic convex hull \cite[Example 2.10 and Proposition 6.2]{LL}.

Suppose a finite group $G$ acts isometrically on $\mathscr P$ and satisfies $\delta(\gamma g)=\delta(g)$ for every $\gamma\in G$. Then every $g$ admits a $G$-fixed metric $\bar g$ with $\delta(\bar g)\leq\delta(g)$. The orbit $Gg$ lies in the closed, geodesically convex sublevel set $\{h:\delta(h)\leq\delta(g)\}$---see Section~\ref{sec:direct}---and the barycenter will belong to the same sublevel set and by uniqueness is $G$-invariant: a global minimizer can be chosen $G$-invariant.

For a fusion category $\mathcal C$, let $G\leq S_n$ act on $\mathcal C^{\boxtimes n}$ by permuting factors. Objects of the equivariantization must carry coherent identifications and have morphisms obey these identifications by \cite[Section 4.15]{EGNO}.

\begin{theorem}\label{sym:equiv}
For every complex fusion category $\mathcal C$ and finite subgroup $G\leq S_n$
\begin{equation}\label{sym:bound}
D_{\mathrm{assoc}}\bigl((\mathcal C^{\boxtimes n})^G\bigr)\leq D_{\mathrm{assoc}}(\mathcal C^{\boxtimes n}).
\end{equation}
\end{theorem}

\begin{proof}
Have an object in the semisimple model as a family of multiplicity spaces $W_a$. Have tensor products take the form
\[
(W\otimes Z)_c=\bigoplus_{a,b}V_{ab}^c\otimes W_a\otimes Z_b.
\]
Start by fixing the invariant elementary forms on $V_{ab}^c$ and positive forms on the multiplicity spaces; the permutation action preserves the former, while averaging over $G$ will makes equivariant structure maps isometric. We can decompose binary tensor products orthogonally into simple equivariant summands.

For each $z$, the induced form on $\operatorname{Hom}(z,x\otimes y)$ can be found with
\[
v^\dagger w=g(v,w)\id_z.
\]

Now, let $c=\delta(g)$. On arbitrary underlying objects, the associator is an orthogonal direct sum of simple associator blocks tensored with identities and its squared singular values will therefore be bounded above by $1+c$ and below by $1-c$ when $c<1$. By naturality and the monoidal action equations, the associator and its inverse are equivariant. Restricting these operator inequalities to the equivariant Hom spaces preserves the bounds, which pass to the equivariant associator blocks under the unitary binary decompositions. Their Gram defects are at most $c$ as a result. When $c\geq1$, the upper bound alone suffices, since the lower singular-value deficit of an invertible block is less than one.

Factor permutations act isometrically on $\mathscr P$ and preserve $\delta$. The barycenter argument will result in an invariant minimizer on $\mathcal C^{\boxtimes n}$ which the preceding construction applies.

\end{proof}

Consider the following:
\[
    \mathcal E_n=(\mathrm{YL}^{\boxtimes n})^{S_n},\qquad n\ge1.
\]

The simple objects are indexed by $0\le k\le n$ and have irreducible representations $\rho$ of $S_k\times S_{n-k}$. Because the canonical permutation action has trivial stabilizer cocycles, so the orbit--stabilizer formula \cite[Corollary 2.13]{BN} gives
\begin{equation}\label{sym:dimensions}
    \operatorname{FPdim}(X_{k,\rho})=\binom{n}{k}(\dim\rho)\phi^k,\qquad
    \dim(X_{k,\rho})=\binom{n}{k}(\dim\rho)(-\phi^{-1})^k.
\end{equation}

The spherical structure are inherited from Yang--Lee.

Here is an example with non-scalar elementary metrics occurs for $n=3$. The category $\mathcal E_3$ has ten simple objects. For $k=0,3$, the stabilizer representations are $\mathbf1,\operatorname{sgn},V$, where $V$ is the two-dimensional standard representation of $S_3$; as for $k=1,2$, they are the two characters $+$ and $-$ of $S_2$. Taking $Y=X_{3,V}$, we find

\begin{equation}\label{sym:fusion}
    \begin{split}
    Y\otimes Y={}&
    \bigoplus_{\rho\in\{\mathbf1,\operatorname{sgn},V\}}
    (X_{0,\rho}\oplus X_{3,\rho})\\
    &{}\oplus2X_{1,+}\oplus2X_{1,-}
    \oplus2X_{2,+}\oplus2X_{2,-}.
    \end{split}
\end{equation}

To see this, expand $(\tau\otimes\tau)^{\boxtimes3}$ over subsets of the three factors. The multiplicity at a subset of size $k$ is given by $\operatorname{Res}_{S_k\times S_{3-k}}^{S_3}(V\otimes V)$. Using
\[
    V\otimes V=\mathbf1\oplus\operatorname{sgn}\oplus V,
    \qquad V|_{S_2}=+\oplus-,
\]
to obtain \eqref{sym:fusion}. We shall note that the multiplicity-two channels admit arbitrary positive Hermitian $2\times2$ metrics, including off-diagonal terms.

We would also want to bound the associator defect of the entire family. Let $M=\phi+\sqrt\phi$ and let $\lambda>1$ satisfy $\lambda+\lambda^{-1}=3\phi+1$. Then,
\begin{equation}\label{sym:familybound}
    (\phi^n+\phi^{n/2})^2
    \le 1+D_{\mathrm{assoc}}(\mathcal E_n)
    \le\lambda^{2\lfloor n/2\rfloor}
    M^{2(n\bmod2)}.
\end{equation}
The upper bound follows directly from Theorem~\ref{sym:equiv} and the tensor-power bound discussed in Section~\ref{sec:square}.

For the lower bound, take $x=X_{n,\mathbf1}$. The channels $\one$ and $x$ in $x\otimes x$ both have multiplicity one. Their entries in $F_{xxx}^{x}$ agree with those of the tensor-power associator, giving the corner
\[
    T=\begin{pmatrix}
    (-\phi)^n&i^n\phi^{n/2}\\
    i^n\phi^{n/2}&\phi^n
    \end{pmatrix}.
\]
For an arbitrary elementary metric, this corner would be transformed by positive diagonal similarity, since the corresponding left and right metrics are $g_{xx}^{\one}$ and $(g_{xx}^{x})^2$. Let $A=\phi^n$ and $B=\phi^{n/2}$. Every conjugate has squared Hilbert--Schmidt norm of at least $2(A^2+B^2)$ and with an absolute determinant $A^2-B^2$, therefore its largest singular value is therefore at least $A+B$. Since compression cannot itself increase the operator norm, we obtain the lower bound in \eqref{sym:familybound}.

\subsection{The center invariant and vacuum amplification}\label{sec:bulk}
Let us now turn to thinking about bulk charges.

The bulk excitations are described by the Drinfeld center $Z(\mathcal C)$. Define
\begin{equation}\label{bulk:eq:bulk}
    D_{\mathrm{bulk}}(\mathcal C)
    :=D_{\mathrm{assoc}}(Z(\mathcal C)).
\end{equation}
The minimization is over elementary metrics in $Z(\mathcal C)$. By choosing this center, the bulk fusion structure is independent of a choice of a Morita-equivalent input. For all complex fusion categories, $D_{\mathrm{bulk}}$ is actually even invariant under categorical Morita equivalence
\begin{equation}\label{bulk:eq:zero}
    D_{\mathrm{bulk}}(\mathcal C)=0
    \quad\Longleftrightarrow\quad
    \mathcal C\text{ is unitarizable}.
\end{equation}
If $\mathcal U$ is unitarizable, then
\begin{equation}\label{bulk:eq:ancilla}
    D_{\mathrm{bulk}}(\mathcal C\boxtimes\mathcal U)
    =D_{\mathrm{bulk}}(\mathcal C).
\end{equation}

As Morita equivalent fusion categories have braided equivalent centers \cite[Theorem 8.12.3]{EGNO}, the first assertion is true. For \eqref{bulk:eq:zero}, recall the Morita equivalence
\[
    Z(\mathcal C)\simeq(\mathcal C\boxtimes\mathcal C^{\mathrm{op}})^*_{\mathcal C}
\]
\cite[Proposition 7.13.8]{EGNO}. Unitarizability is Morita invariant \cite[Corollary 5.13]{CFHPS}, and a unitary structure on $\mathcal C\boxtimes\mathcal C^{\mathrm{op}}$ will restrict it to its full fusion subcategory $\mathcal C\boxtimes\mathrm{Vec}$. Conversely, the center of a unitarizable fusion category is unitarizable \cite{Galindo}.

Finally,
\[
    Z(\mathcal C\boxtimes\mathcal U)\simeq
    Z(\mathcal C)\boxtimes Z(\mathcal U).
\]
For any fusion category $\mathcal A$ and unitarizable $\mathcal V$, product metrics give
$D_{\mathrm{assoc}}(\mathcal A\boxtimes\mathcal V)\leq D_{\mathrm{assoc}}(\mathcal A)$, while restriction to $\mathcal A\boxtimes\mathrm{Vec}$ also gives the reverse inequality. Apply this with $\mathcal V=Z(\mathcal U)$ to prove \eqref{bulk:eq:ancilla}.

We can obtain a lower bound from the vacuum coefficients of the associators. Let $i$ be a simple object with dual $\bar i$. The intrinsic squared norm $|i|^2$ is positive \cite[Section 7.21]{EGNO} equals $q_i^2$ in a spherical structure with dimension $q_i$ \cite[Proposition 7.21.14]{EGNO}.

The vacuum channels of $F_{i\bar i i}^{i}$ are one-dimensional.
\begin{equation}\label{bulk:eq:vacuum}
    a_i=[F_{i\bar i i}^{i}]_{1,1}.
\end{equation}
Under an elementary gauge $S$, normalized on unit fusion spaces,
\begin{equation}\label{bulk:eq:covariance}
    a_i(S\cdot F)
    =\frac{S_{\bar i i}^{1}}{S_{i\bar i}^{1}}a_i(F).
\end{equation}

Let us choose vacuum inclusions $b_i:\one\to i\otimes\bar i$ and projections $p_i:i\otimes\bar i\to\one$ for $p_i b_i=1$.
\[
    (\operatorname{id}_i\otimes p_{\bar i})\alpha_{i,\bar i,i}(b_i\otimes\operatorname{id}_i)=a_i\operatorname{id}_i.
\]
By rigidity, $a_i\neq0$, and $\operatorname{coev}_i=b_i$, $\operatorname{ev}_i=a_i^{-1}p_{\bar i}$ will satisfy the triangle identities. Making the same choice for $\bar i$ and identifying the double dual of $i$ with $i$, the two traces defining $|i|^2$ are $a_{\bar i}^{-1}$ and $a_i^{-1}$. This proves \eqref{bulk:eq:vacuum-product}.
\begin{equation}\label{bulk:eq:vacuum-product}
    a_i a_{\bar i}=|i|^{-2}.
\end{equation}
The lower bound is
\begin{equation}\label{bulk:eq:V}
    V(\mathcal C)=\max_{i\text{ simple}}\frac1{|i|^2}
    =\frac1{\min_{i\text{ simple}}|i|^2}.
\end{equation}
Clearly $V(\mathcal C)\geq1$ since $|\one|^2=1$.
\begin{theorem}\label{bulk:statement:3}
For every complex fusion category and every integer $n\geq1$
\begin{equation}\label{bulk:eq:vacuum-bound}
    1+D_{\mathrm{assoc}}(\mathcal C^{\boxtimes n})
    \geq V(\mathcal C)^n.
\end{equation}
\begin{equation}\label{bulk:eq:Vproduct}
    V(\mathcal C\boxtimes\mathcal D)
    =V(\mathcal C)V(\mathcal D).
\end{equation}
\end{theorem}

\begin{proof}
Fix a coherent metric $g$ and write
\[
    M=\max_B\|U_B(g)\|_{\mathrm{op}},\qquad
    \delta=\max_B\|U_B(g)^\dagger U_B(g)-I\|_{\mathrm{op}}.
\]
It is clear that every matrix entry in the orthonormal tree bases is bounded by $M$. Applying this fact to the two vacuum coefficients and using \eqref{bulk:eq:vacuum-product}, we obtain
\[
    M^2\geq|a_i a_{\bar i}|=|i|^{-2},
    \qquad 1+\delta\geq M^2.
\]
By maximizing over $i$ and minimizing over the metrics, we get the first bound.
Because norms are multiplicative under Deligne products, we obtain \eqref{bulk:eq:Vproduct}. Applying the first bound to $\mathcal C^{\boxtimes n}$ gives \eqref{bulk:eq:vacuum-bound}.

\end{proof}

When $V(\mathcal C)>1$, the vacuum bound is strict:
\begin{equation}\label{bulk:eq:strict-vacuum}
    1+D_{\mathrm{assoc}}(\mathcal C)>V(\mathcal C).
\end{equation}
The second triangle identity for $\bar i$ gives
\[
    [(F_{i\bar i i}^{i})^{-1}]_{1,1}=a_{\bar i}.
\]

Now suppose $1+D_{\mathrm{assoc}}(\mathcal C)=V(\mathcal C)$. We can choose a minimizing metric, which we know exists by Section~\ref{sec:direct}, and a simple $i$ with $|i|^{-2}=V(\mathcal C)$. Every forward block must have a norm at most $\sqrt{V(\mathcal C)}$; since $|a_i a_{\bar i}|=V(\mathcal C)$, both vacuum entries must have modulus $\sqrt{V(\mathcal C)}$. The vacuum entry of $F_{i\bar i i}^{i}$ therefore saturates the column norm, which forces all other entries in that column to vanish. Invertibility gives $[(F_{i\bar i i}^{i})^{-1}]_{1,1}=a_i^{-1}$, and hence $a_i a_{\bar i}=1$, a contradiction.

\begin{theorem}\label{bulk:statement:4}
    Let $\mathcal C$ be a spherical complex fusion category and set $r(\mathcal C)=\operatorname{FPdim}(\mathcal C)/\dim(\mathcal C)$. For every $n\geq1$,
\begin{equation}\label{bulk:eq:main}
    D_{\mathrm{bulk}}(\mathcal C^{\boxtimes n})
    \geq r(\mathcal C)^{2n}-1.
\end{equation}
The Morita-invariant quantity
\begin{equation}\label{bulk:eq:Vbulk}
    V_{\mathrm{bulk}}(\mathcal C):=V(Z(\mathcal C))
\end{equation}
will satisfy
\begin{align}
    1+D_{\mathrm{bulk}}(\mathcal C^{\boxtimes n})
    &\geq V_{\mathrm{bulk}}(\mathcal C)^n
    =r(\mathcal C)^{2n},\label{bulk:eq:refined}\\
    V_{\mathrm{bulk}}(\mathcal C)=1
    &\quad\Longleftrightarrow\quad
    \mathcal C\text{ is pseudo-unitary}.\label{bulk:eq:pseudo}
\end{align}
The inequality in \eqref{bulk:eq:refined} is strict if  $\mathcal C$ is not pseudo-unitary.
\end{theorem}

\begin{proof}
Applying \eqref{bulk:eq:exact-V} to $Z(\mathcal C)$ gives $V_{\mathrm{bulk}}(\mathcal C)=r(\mathcal C)^2$. The equivalence $Z(\mathcal C^{\boxtimes n})\simeq Z(\mathcal C)^{\boxtimes n}$ implies Theorem~\ref{bulk:statement:3}

If $\mathcal C$ is pseudo-unitary, its center must also be so. Every simple squared norm in the center is then the square of a Frobenius--Perron dimension---at least one---with equality for the unit. The converse follows from $V_{\mathrm{bulk}}=r^2$. Morita invariance follows from equivalence of centers, and multiplicativity from \eqref{bulk:eq:Vproduct}. If $\mathcal C$ is not pseudo-unitary, then $V_{\mathrm{bulk}}>1$, strictness follows from \eqref{bulk:eq:strict-vacuum}.

\end{proof}

The product inequality makes $\log(1+D_{\mathrm{bulk}}(\mathcal C^{\boxtimes n}))$ subadditive. Thus
\begin{equation}\label{bulk:eq:rate}
    \rho_{\mathrm{bulk}}(\mathcal C)
    =\lim_{n\to\infty}
    (1+D_{\mathrm{bulk}}(\mathcal C^{\boxtimes n}))^{1/n}
    =\inf_{n\geq1}
    (1+D_{\mathrm{bulk}}(\mathcal C^{\boxtimes n}))^{1/n}
\end{equation}
exists by subadditivity (see Appendix~\ref{app:subadditivity}).
The growth rate is Morita invariant and satisfies
\[
    \rho_{\mathrm{bulk}}(\mathcal C^{\boxtimes k})
    =\rho_{\mathrm{bulk}}(\mathcal C)^k.
\]

Theorem~\ref{bulk:statement:4} gives
\begin{equation}\label{bulk:eq:ratebound}
    \rho_{\mathrm{bulk}}(\mathcal C)
    \geq V_{\mathrm{bulk}}(\mathcal C)
    =r(\mathcal C)^2.
\end{equation}
In particular, for spherical categories,
\begin{equation}\label{bulk:eq:chain}
    \mathcal C\text{ unitarizable}
    \ \Longrightarrow\ \rho_{\mathrm{bulk}}(\mathcal C)=1
    \ \Longrightarrow\ \mathcal C\text{ pseudo-unitary}.
\end{equation}
The strict finite-copy bound does not necessarily imply strictness in \eqref{bulk:eq:ratebound} as the excess could be sublinear in $n$.

Back to the Yang--Lee example, the dimensions are
\[
    \dim(\mathrm{YL})=1+\phi^{-2},\qquad
    \operatorname{FPdim}(\mathrm{YL})=1+\phi^2,\qquad
    r(\mathrm{YL})=\phi^2.
\]
Yang--Lee has a nondegenerate braiding \cite[Example 8.18.5 and Exercise 8.18.7]{EGNO}, so $Z(\mathrm{YL})\simeq\mathrm{YL}\boxtimes\mathrm{YL}^{\mathrm{rev}}$ \cite[Proposition 8.20.12]{EGNO}, where $\mathrm{rev}$ reverses the braiding. By Section~\ref{sec:square}, we obtain
\begin{equation}\label{bulk:eq:YLexact}
    D_{\mathrm{bulk}}(\mathrm{YL})
    =\lambda^2-1,\qquad
    \lambda=\frac{3\phi+1+\sqrt{(3\phi+1)^2-4}}2.
\end{equation}
In comparison,
\[
    V_{\mathrm{bulk}}(\mathrm{YL})-1=\phi^4-1
    <D_{\mathrm{bulk}}(\mathrm{YL}).
\]

Applying the tensor-corner bound of Section~\ref{sec:square} to $2n$ Yang--Lee factors gives
\begin{equation}\label{bulk:eq:YLpowers}
    (\phi^{2n}+\phi^n)^2
    \leq1+D_{\mathrm{bulk}}(\mathrm{YL}^{\boxtimes n})
    \leq\lambda^{2n}.
\end{equation}

\[
    \phi^4\leq\rho_{\mathrm{bulk}}(\mathrm{YL})\leq\lambda^2.
\]

We also obtain an unbounded family from the modular $\mathrm{SU}(2)$ fusion categories $\mathcal C_{h,m}$ at level $h-2$, with $q=e^{m\pi i/h}$, $h\geq3$, and $\gcd(m,2h)=1$. Their simple dimensions and Frobenius--Perron dimensions are \cite[Example 8.18.5]{EGNO}
\[
    q_j=\frac{\sin((j+1)m\pi/h)}{\sin(m\pi/h)},\qquad
    d_j=\frac{\sin((j+1)\pi/h)}{\sin(\pi/h)},
    \qquad 0\leq j\leq h-2.
\]
If we sum their squares (Appendix~\ref{app:dimensions}), we get
\[
    r(\mathcal C_{h,m})
    =\frac{\sin^2(m\pi/h)}{\sin^2(\pi/h)},\qquad
    D_{\mathrm{bulk}}(\mathcal C_{h,m}^{\boxtimes n})
    \geq\left(\frac{\sin^2(m\pi/h)}{\sin^2(\pi/h)}\right)^{2n}-1.
\]
For odd $h$, take $m$ to be the odd integer among $(h-1)/2$ and $(h+1)/2$. Then
\begin{equation}\label{bulk:eq:quartic}
    D_{\mathrm{bulk}}(\mathcal C_{h,m})
    \geq\frac1{16\sin^4(\pi/(2h))}-1
    =\frac{h^4}{\pi^4}+O(h^2).
\end{equation}

By the strict vacuum bound, $E_{\mathrm{bulk}}(\mathcal C)\geq0$. Equality is true if and only if $\mathcal C$ is unitarizable.
Now, for spherical $\mathcal C$, define
\begin{equation}\label{bulk:eq:excess}
    E_{\mathrm{bulk}}(\mathcal C)
    =\log\frac{1+D_{\mathrm{bulk}}(\mathcal C)}{r(\mathcal C)^2}.
\end{equation}

For pseudo-unitary categories, the vacuum bound is trivial. This leaves the following question.

\begin{question}\label{question:pseudo}
Does every pseudo-unitary complex fusion category admit a unitary structure? Equivalently, does
\begin{equation}\label{bulk:eq:target}
    V_{\mathrm{bulk}}(\mathcal C)=1
    \quad\Longrightarrow\quad
    D_{\mathrm{bulk}}(\mathcal C)=0
\end{equation}
hold?
\end{question}

\subsection{Quantum-group families with fusion multiplicities}\label{sec:quantumfamily}
We will now give a family with adjoint fusion multiplicity two, for which the vacuum bound is independent of the multiplicity metrics.
Let $L\ge7$ be odd and put $k=L-3$. Choose the odd integer $t_L$ among $(L-1)/2$ and $(L+1)/2$. The Galois automorphism
\[
    e^{\pi i/L}\longmapsto e^{t_L\pi i/L}
\]
acts on the algebraic tensor and spherical data of the unitary $SU(3)_k$ fusion category, and denote the resulting category by $\mathcal Q_L$. The choice of exponent is further discussed in Appendix~\ref{app:dimensions}.

The simple objects are labeled by $(a,b)$ with $a,b\ge0$ and $a+b\le k$, giving $(L-2)(L-1)/2$ simples. Their quantum dimensions are given by the quantum Weyl formula \cite[Section 3.1]{Rowell}:
\begin{equation}\label{quant:qdim}
    [m]_t=\frac{\sin(mt\pi/L)}{\sin(t\pi/L)},\qquad
    \dim_t(a,b)=\frac{[a+1]_t[b+1]_t[a+b+2]_t}{[2]_t}.
\end{equation}
The Frobenius--Perron dimensions are the positive values at $t=1$.

Let $x=(1,1)$ be the self-dual adjoint object. For $k\ge4$, its tensor square is
\begin{equation}\label{quant:fusion}
    x\otimes x=\one\oplus2x\oplus(3,0)\oplus(0,3)\oplus(2,2).
\end{equation}
 The adjoint multiplicity is also discussed in \cite[Section 5.2 and Eq. (33)]{AS}. Thus
\[
    \dim V_{xx}^{x}=2,\qquad
    \dim\operatorname{Hom}(x,x^{\otimes3})=1+4+1+1+1=8.
\]
The intermediate channel $x$ gives us four dimensions, with tree metric $g_{xx}^{x}\otimes g_{xx}^{x}$ for an arbitrary positive Hermitian $2\times2$ matrix $g_{xx}^{x}$.

Put $c_L=\cos(\pi/L)$. Since $t_L\pi/L=\pi/2\pm\pi/(2L)$, we have
\begin{equation}\label{quant:adjdim}
    \dim_{t_L}(x)=[2]_{t_L}[4]_{t_L}
    =-8\sin^2\!\left(\frac\pi{2L}\right)\cos\!\left(\frac\pi L\right)
    =-4c_L(1-c_L).
\end{equation}
Since $1/2<c_L<1$, the intrinsic squared norm of $x$ must be less than one. Set
\[
    v_L=|x|^{-2}=\frac1{16c_L^2(1-c_L)^2}>1.
\]

We can improve the vacuum bound for any self-dual simple $x$ with $v=|x|^{-2}>1$. Let $d_x=\operatorname{FPdim}(x)$, $\Phi=\operatorname{FPdim}(\mathcal C)$, and let
\[
    J_x=\Phi^3/d_x^4-1.
\]

The vacuum identities imply $J_x\ge1$.

\begin{theorem}\label{quant:gap}
The equation
\begin{equation}\label{quant:root}
(s-v)(s^{J_x}-v)=(v-1)^2
\end{equation}
uniquely has the solution $s_*(v,J_x)>v$. Every coherent elementary metric spaces has a maximum Gram defect of at least $s_*(v,J_x)-1$. Particularly,
\begin{equation}\label{quant:gapbound}
D_{\mathrm{assoc}}(\mathcal C)\ge s_*(v,J_x)-1>v-1.
\end{equation}
\end{theorem}

\begin{proof}
Fix a coherent metric. Let $M=\max_B\|U_B\|$, and let $F$ for the normalized $(x,x,x;x)$ block. The vacuum identities of Section~\ref{sec:bulk} means that
\[
F_{00}=(F^{-1})_{00}=a,\qquad a^2=v.
\]
Let
\[
Fe_0=ae_0+z,\qquad e_0^\dagger F^{-1}=ae_0^\dagger+w^\dagger,\qquad z,w\perp e_0.
\]
It follows that $1=v+w^\dagger z$. The forward and inverse norm bounds give
\[
    \|z\|^2\le M^2-v,\qquad
    \|w\|^2\le M^{2J_x}-v.
\]
By Cauchy--Schwarz,
\[
    (v-1)^2\le(M^2-v)(M^{2J_x}-v).
\]
The right-hand side is strictly increasing as a function of $M^2>v$ as long as it is greater than zero. As a result $M^2\ge s_*(v,J_x)$. Since $1+\delta(g)\ge M^2$, taking the infimum will prove \eqref{quant:gapbound}.
\end{proof}

For the adjoint of $\mathcal Q_L$, we have the minimum and
\begin{equation}\label{quant:growth}
    1+D_{\mathrm{assoc}}(\mathcal Q_L)
    \ge s_*(v_L,J_x)>v_L>\frac{L^4}{4\pi^4}.
\end{equation}
The last inequality naturally follows since $1-\cos u<u^2/2$ and $c_L<1$. Every $\mathcal Q_L$ is therefore nonunitarizable.

The exponent $J_x$ can be computed from
\[
d_x=[2]_1[4]_1,\qquad \Phi=\sum_{a+b\le L-3}\left(\frac{[a+1]_1[b+1]_1[a+b+2]_1}{[2]_1}\right)^2.
\]
For $L=7$, the category already has fifteen simple objects and an eight-dimensional adjoint recoupling block; the same vacuum coefficient gives a measurement obstruction. On $\operatorname{Hom}(x,x^{\otimes3})$, let $P$ and $Q$ pick vacuum fusions of the first and last pairs respectively. These rank-one idempotents satisfy
\begin{equation}\label{quant:overlap}
PQP=v_LP.
\end{equation}

Since $v_L>1$, they cannot both be orthogonal projection in a common positive metric. Given a faithful encoding $A$ and physical effect $E_1,E_2$, the two-selector bound of Section~\ref{sec:measurements} would imply that
\begin{equation}\label{quant:effect}
    \max_{j=1,2}\|(E_jA-AP_j)(A^\dagger A)^{-1/2}\|
    \ge\frac{\sqrt{v_L}+\sqrt{v_L-1}-1}{2},
    \qquad(P_1,P_2)=(P,Q).
\end{equation}

%% file: 04-local-fusion.tex
\label{part:locality}\label{sec:charged}

A string-net construction tells us how to evaluate fusion diagrams. But in order to interpret these diagrams as operations on physical states, however, we would also like to specify charged Hilbert spaces and operators acting on their charges. It can be useful to think of this as a question: can a fusion vertex is used can the physical overlaps of these operations recover the same positive elementary form each time?

\subsection{String nets and charged realizations}\label{subsec:chargedrealizations}

Recall the idempotent example from the introduction for a second. Of course a selector may satisfy $B_p^2=B_p$ without being an orthogonal projection. What does this mean for the physical realization of a string-net model? For spherical nonunitary data, there still is idempotence by graphical calculations, but not necessarily self-adjointness, which we wanted.

Consider an oriented trivalent cellulation of a surface with a choice of branching at each possible vertex. Each edge carries a simple label, replaced by its dual when the orientation is reversed, while a vertex with incoming labels $a,b$ and outgoing label $c$ carries a vector in $V_{ab}^c$. For fixed boundary labels, the space of admissible configurations is a direct sum over internal edge labels of tensor products of these fusion spaces \cite[Sections 2--3]{Kirillov}. Choosing bases and declaring the configurations orthonormal gives a Hilbert space. A problem arises when we change the fusion tree, since recoupling might not preserve this inner product. On a three-input, one-output patch, the failure is measured by $U_B^\dagger U_B-I$.

Now consider the operators of the string-net Hamiltonian. The vertex constraint $Q_v$ only checks whether the labels are admissible, so it is an orthogonal projection in the chosen basis. The plaquette operator $B_p$, on the other hand, acts by inserting a weighted sum of labeled loops and reducing the resulting diagram using fusion, associators, and duality. In the unitary Levin--Wen model, this construction will give us a Hermitian projection \cite[Eqs. (10)--(14) and Appendix C]{LW}. The graphical calculation of $B_p^2=B_p$ continues to hold for spherical nonunitary categories. But this alone tells us  nothing about whether $B_p$ is self-adjoint, which instead depends on the inner product we have chosen for the Hilbert space.

\begin{figure}[tbp]
\centering
\begin{tikzpicture}[x=1cm,y=1cm,font=\small,line width=.75pt]
\foreach \xx in {-3.3,3.3} {
  \begin{scope}[shift={(\xx,0)}]
  \draw (0:1.6)--(60:1.6)--(120:1.6)--(180:1.6)--(240:1.6)--(300:1.6)--cycle;
  \foreach \aa in {0,60,120,180,240,300} {\draw (\aa:1.6)--(\aa:2.1);}
  \draw[diagramblue,line width=1pt] (0,0) circle(.85);
  \node[diagramblue] at (.35,.45) {$\Omega$};
  \end{scope}
}
\draw[diagramblue,line width=1pt] (-3.3,0) circle(1.1);
\node[diagramblue] at (-3.3,-1.25) {$\Omega$};
\node at (0,0) {$=$};
\node at (-3.3,-2.5) {$B_p^2$};
\node at (3.3,-2.5) {$B_p$};
\end{tikzpicture}
\caption{Vacuum-loop insertion on a plaquette.}
\label{fig:plaquette-loop}
\end{figure}

Let us first see what remains of the usual Levin--Wen construction when the category is not unitary. For a spherical fusion category $\mathcal A$, we will denote by $L_p^s$ the operator which inserts a loop labeled by $s$ around a plaquette $p$. The loops are reduced using the associators of $\mathcal A$ given that the branching and multiplicity bases fixed. In the unrescaled convention, these operators satisfy the usual fusion law
\begin{equation}\label{sn:looplaw}
    L_p^sL_p^t=\sum_uN_{st}^uL_p^u,\qquad L_p^1=I.
\end{equation}
We can therefore define the vacuum plaquette operator in the same way as in the unitary case. Writing $q_s=\dim(s)$ and $\Delta=\sum_s q_s^2$ for the global dimension, we have
\begin{equation}\label{sn:plaquette}
    B_p=\frac1\Delta\sum_s q_sL_p^s.
\end{equation}

The first thing to notice is that $B_p$ is still idempotent. Indeed, the spherical dimensions satisfy $q_{s^*}=q_s$, and together with the fusion-character identity and Frobenius reciprocity, this gives
\[
    \sum_tN_{st}^u q_t=q_sq_u,
    \qquad \sum_{s,t}q_sq_tN_{st}^u=\Delta q_u.
\]
It follows that
\begin{equation}\label{sn:idempotence}
    B_p^2=\frac1{\Delta^2}\sum_u\left(\sum_{s,t}q_sq_tN_{st}^u\right)L_p^u=B_p.
\end{equation}
Graphical arguments using disjoint loops and the pentagon identity will show that neighboring plaquette operators commute \cite[Theorem 3.9]{Kirillov}. If we also include the vertex constraints $Q_v$, we also obtain a family of commuting algebraic idempotents. We write the Hamiltonian
\begin{equation}\label{sn:Halg}
    H_{\mathrm{alg}}=\sum_v(I-Q_v)+\sum_p(I-B_p),
\end{equation}
The main difference to the unitary case is that these operators here need not be Hermitian.

The decomposition $V=\ran B\oplus\ker B$ holds algebraically. It is orthogonal, however, only when $B^\dagger=B$. For a finite lattice, one can get around this by always choosing an inner product which makes all the commuting idempotents orthogonal. The difficulty in this is that such a choice might not respect locality when the lattice is enlarged or if additional spectator charges are introduced. We discuss this issue in Appendix~\ref{app:lifts}.

To describe the bulk excitations,  first recall that an object $x$ must come with a prescription for passing it across any input string $a$. This is given by a half-braiding $c_{x,a}:x\otimes a\to a\otimes x$, which is natural in $a$ and compatible with tensor products. The pairs $(x,c_x)$ form the Drinfeld center $Z(\mathcal A)$. The corresponding charge sectors can also be described by annular idempotents in the tube algebra or MPO formalism \cite[Section 5.1]{MPO}; also see \cite[Theorem 6.4]{Kirillov}. For these idempotents to represent physical charge measurements, however, their ranges must be orthogonal in the physical Hilbert space. This is discussed in Appendix~\ref{app:lifts}.

In what follows, $\mathcal C$ denotes the represented charge category, with  us adopting the notation of $\mathcal C=Z(\mathcal A)$ when only bulk charges are considered.

\subsection{Constructing charged operations}\label{subsec:lifts}
Suppose we try to split a charge into two charges in some small neighborhood. In general, there may be other charges far away from this neighborhood, which we will call spectators. Whilst these spectators do not take part in the splitting, they do affect the global code space, so we would much prefer the local operation to represent the same fusion vector regardless of their presence---or lack thereof.

Let $\C$ be a complex fusion category with fusion spaces $V_{ab}^{c}=\operatorname{Hom}(c,a\otimes b)$ and exact normalized associator blocks $F_B$.
\[
    m=\max_{a,b,c}\dim V_{ab}^{c},\qquad M=\max_{a,b}\sum_c\dim V_{ab}^{c},
\]
where both numbers are finite and $m\ge1$. At a separation scale $\ell$, we use $s$ to denote a spectator context, which records the remote charges and positions, and any outer charge constraint---we will include every constant that can appear in a two-vertex fusion tree. For each vertex type $p=(a,b;c)$, suppose we have a linear family of splitting maps
\[
    t_{p,s}(v):\K_{c,s}\longrightarrow\K_{ab,s},\qquad v\in V_{ab}^{c},
\]
between nonzero finite-dimensional code spaces, which we obtain by restricting local physical operators $\widehat t_{p,s}(v)$ which preserve their corresponding codes. The spaces and operators depends on $\ell$. In particular, the products of these maps satisfy
\begin{equation}\label{charged:eq:compression}
    t_{p,s}(v)^\dagger t_{p,s}(w)=P_{c,s}\widehat t_{p,s}(v)^\dagger\widehat t_{p,s}(w)P_{c,s}
\end{equation}
on $\K_{c,s}$, where the adjoint is taken with respect to the physical inner product.

As for a tensor-network model, one might try to construct these operators directly from the virtual fusion spaces. Consider a PEPS patch $R$, with a map $A_R$ from their boundary virtual boundary support into the physical Hilbert space, and let $X(v)$ be a virtual splitting map. We want to find a physical operator $E_R(v)$ such that
\begin{equation}\label{charged:eq:lift-peps}
    E_R(v)A_R=A_{R'}X(v)
\end{equation}
on this support, where the two regions are placed in a common patch with only suitable ancillary degrees of freedom if their physical site sets do differ.

If two virtual vectors represent the same physical state, when we apply the splitting map, we must get the same physical result. This requires $A_{R'}X(v)$ to vanish on $\ker A_R$, in which case we can take
\[
    E_R(v)=A_{R'}X(v)A_R^+,
\]
where $A_R^+$ is the Moore--Penrose inverse. The construction is discussed in Appendix~\ref{app:lifts}.

Very unfortunately however, there is one further issue. The virtual boundary support of a PEPS patch may depend on its environment. One should notice that adding spectator charges can change the support. Because the boundary supports must be compatible under gluing, the same local operations must extend to the global charged states. The corresponding conditions and their relation to PEPS parent Hamiltonians are discussed in Appendix~\ref{app:realization-background}.

\subsection{Local fusion and the recoupling obstruction}\label{subsec:localmaps}
Let us first consider the case where the family of maps have physical overlaps that agree with the elementary forms. For a splitting of $p=(a,b;c)$, suppose that, given $g_p$ as independent of the spectator context $s$,
\[
t_{p,s}(v)^\dagger t_{p,s}(w)=g_p(v,w)I.
\]
This form is positive definite if the splitting maps are faithful because every nonzero fusion vector then acts nontrivially on some incoming state.

Now let us consider two successive splittings: from $d$ to $e,c$; from $e$ to $a,b$. For $v,v'\in V_{ab}^e$ and $w,w'\in V_{ec}^d$, apply the same scalar identity twice to obtain
\[
\begin{aligned}
&\langle t_{(a,b;e),s(c)}(v)t_{(e,c;d),s}(w)\psi,\, t_{(a,b;e),s(c)}(v')t_{(e,c;d),s}(w')\chi \rangle\\
&\hspace{2em}=g_{ab}^e(v,v')\,g_{ec}^d(w,w')\langle\psi,\chi\rangle.
\end{aligned}
\]
Notice that whilst $c$ becomes a spectator in the second splitting, $g_{ab}^e$ is still unchanged by assumption. Since different intermediate charges are orthogonal,
\[
G_L=\bigoplus_e g_{ab}^e\otimes g_{ec}^d.
\]
Repeating the argument for the other fusion tree,
\[
G_R=\bigoplus_f g_{bc}^f\otimes g_{af}^d.
\]

After choosing orthonormal coordinates for these elementary forms, each splitting will become an isometry. The same is true for the two-layer tree encodings.

Let $W_L$ and $W_R$ denote the two isometric encodings into a common physical patch. If the physical recoupling agrees with the associator of $\C$, we have $W_L=W_RU_B(g)$, and hence
\[
    I=W_L^\dagger W_L=U_B(g)^\dagger U_B(g).
\]

One can therefore observe that the same elementary forms work for every fusion block, the associator defect must vanish, $D_{\mathrm{assoc}}(\mathcal C)=0$.

Of course, the physical splitting operators will often not satisfy these relations exactly, which could be troublesome. As such, we have to allow errors in their overlaps and a possible difference between the two tree encodings among other potential flaws in order to be able to get more use out of this idea. We also need to control the size of the local operators relative to their action on the code, since otherwise a very small absolute error might simply mean that the splitting itself is becoming small.

To understand how these errors affect the physical realization, we will need to keep track of several quantities. We denote the overlap, leakage, and recoupling errors by $q_\ell$, $\chi_\ell$, and $r_\ell$, respectively. The parameter $K_\ell$ measures the operator norm cost, while $m$ and $M$ are the fusion multiplicity bounds.

\paragraph{H1. Faithfulness and common local expectations.}
For every $p$, choose a reference context $s_0$ in which $v\mapsto t_{p,s_0}(v)$ is injective. Define the positive definite form to be
\begin{equation}\label{charged:eq:g}
    g_p(v,w)=\frac{1}{\dim\K_{c,s_0}}\Tr\!\left(t_{p,s_0}(v)^\dagger t_{p,s_0}(w)\right).
\end{equation}
Unit vertices have form $1$. For every $p,s,v,w$,
\begin{equation}\label{charged:eq:code}
    \norm{t_{p,s}(v)^\dagger t_{p,s}(w)-g_p(v,w)I}
    \le q_\ell\norm{\widehat t_{p,s}(v)}\norm{\widehat t_{p,s}(w)}.
\end{equation}

\paragraph{H2. Norm control of local implementations.}
For all $p,s,v$,
\begin{equation}\label{charged:eq:lift}
    \norm{\widehat t_{p,s}(v)}\le K_\ell\norm{t_{p,s}(v)},
\end{equation}
holds, where $K_\ell\ge1$ is uniform. Also, set
\begin{equation}\label{charged:eq:epsilon}
    \eps_\ell=q_\ell K_\ell^2.
\end{equation}

\paragraph{H3. Physical charge resolution.}
Mutually orthogonal charge projections $\Pi_{c,s}$ on each common target space exist, which satisfy
\begin{equation}\label{charged:eq:charge}
    \norm{(I-\Pi_{c,s})t_{ab,s}^{c}(v)}
    \le\chi_\ell\norm{t_{ab,s}^{c}(v)}.
\end{equation}

\paragraph{H4. Composition and the full associator.}
For each associator block, assume that all fusion channels are represented by two composable splitting trees with a common root code $\K_{d,s}$ and physical target, up to a unitary identification $Q_B$. Let $\widetilde W_{L,B}$ and $\widetilde W_{R,B}$ be the corresponding tree maps, orthonormalized using \eqref{charged:eq:g}. We require
\begin{equation}\label{charged:eq:recouple}
    \norm{Q_B\widetilde W_{L,B}-\widetilde W_{R,B}(U_B\otimes I_{\K_{d,s}})}\le r_\ell.
\end{equation}

\subsection{The quantitative obstruction}\label{subsec:quantitative}
We can now reduce the physical question to comparing two tree encodings.

Figure~\ref{fig:physical-comparison} organizes the comparison of the two normalized tree maps.

\begin{figure}[tbp]
\centering
\begin{tikzpicture}[x=1cm,y=1cm,>=Stealth,font=\small,line width=.8pt]
\node[align=center] (L) at (-3.6,1.3) {$V_{L,B}\otimes\mathcal K$\\[4pt]normalized left tree};
\node[align=center] (R) at (3.6,1.3) {$V_{R,B}\otimes\mathcal K$\\[4pt]normalized right tree};
\node[draw,rounded corners=2pt,inner sep=9pt] (P) at (0,-1.7) {common physical patch};
\draw[->] (L)--node[above] {$U_B\otimes I$} (R);
\draw[->,diagramblue] (L)--node[below left] {$Q_B\widetilde W_L$} (P);
\draw[->,diagramblue] (R)--node[below right] {$\widetilde W_R$} (P);
\node at (0,.2) {recoupling error $\le r_\ell$};
\node at (0,-2.5) {tree singular values in $[1-d_\ell,1+d_\ell]$};
\end{tikzpicture}
\caption{The physical comparison behind the local obstruction.}
\label{fig:physical-comparison}
\end{figure}

\begin{theorem}\label{charged:thm:main}
Under H1--H4, if $\eps_\ell<1$ and $d_\ell<1$, then
\begin{equation}\label{charged:eq:main}
    \Def(\C)\le
    \left(\frac{1+d_\ell+r_\ell}{1-d_\ell}\right)^2-1,
    \qquad
    d_\ell=\frac{m q_\ell K_\ell^2+2M\chi_\ell}{1-q_\ell K_\ell^2}.
\end{equation}
Consequently, $\C$ is unitarizable if a sequence of such realizations satisfies
\[
    q_\ell K_\ell^2\longrightarrow0,\qquad
    \chi_\ell\longrightarrow0,\qquad r_\ell\longrightarrow0,
\]
.
\end{theorem}

\begin{proof}
It is easier to break this up. Let us begin with only a single splitting. Set $\eps=\eps_\ell$. From H1 and H2, we obtain
\begin{equation}\label{charged:eq:relative}
    \norm{t_{p,s}(v)^\dagger t_{p,s}(w)-g_p(v,w)I}
    \le\eps\norm{t_{p,s}(v)}\norm{t_{p,s}(w)}.
\end{equation}
Now, choose a $g_p$-orthonormal basis $v_1,\ldots,v_n$ and write $\widetilde t_\alpha=t(v_\alpha)$. For $g_p(v,v)=1$, taking $v=w$ in \eqref{charged:eq:relative} gives $\norm{t(v)}^2\le1+\eps\norm{t(v)}^2$. Note that because $\eps<1$,
\begin{equation}\label{charged:eq:entry}
    \norm{\widetilde t_\alpha}\le(1-\eps)^{-1/2},\qquad
    \norm{\widetilde t_\alpha^\dagger\widetilde t_\beta-\delta_{\alpha\beta}I}
    \le\frac{\eps}{1-\eps}.
\end{equation}

Now, to put it together, combine different fusion channels into one splitting layer. Let $\widetilde T_{ab,s}$ be the column map with columns $\widetilde t_{ab,s}^{c,\alpha}$ and domain $\bigoplus_c\mathbb C^{N_{ab}^{c}}\otimes\K_{c,s}$. For a fixed charge, \eqref{charged:eq:entry} will bound every Gram entry by $\eps/(1-\eps)$ with at most $m$ entries in each row. For distinct charges $c,c'$, the orthogonality $\Pi_c\Pi_{c'}=0$ and H3 will give
\[
    \norm{t_c(v)^\dagger t_{c'}(w)}
    \le2\chi_\ell\norm{t_c(v)}\norm{t_{c'}(w)}.
\]

There are at most $M$ channels, and the row and column sums of the Gram error are therefore bounded by $(m\eps+2M\chi_\ell)/(1-\eps)$. The row--column norm estimate gives
\begin{equation}\label{charged:eq:layer}
    \norm{\widetilde T_{ab,s}^\dagger\widetilde T_{ab,s}-I}\le d_\ell,
    \qquad d_\ell=\frac{m\eps_\ell+2M\chi_\ell}{1-\eps_\ell}.
\end{equation}

We now consider a case where there are two successive splittings. Let $\widetilde W_{L,B}$ and $\widetilde W_{R,B}$ be the normalized tree maps from H4, each obtained by composing two splitting layers in the corresponding spectator contexts. Since the Gram operator of each layer lies between $(1-d_\ell)I$ and $(1+d_\ell)I$, applying these bounds twice gives us
\begin{equation}\label{charged:eq:tree}
    (1-d_\ell)^2I\preceq\widetilde W_{L,B}^\dagger\widetilde W_{L,B}\preceq(1+d_\ell)^2I,
\end{equation}
and similarly for the right tree.

Now, compare the two normalized trees. Let $a=1-d_\ell$, $b=1+d_\ell$, and $r=r_\ell$. By \eqref{charged:eq:tree}, since both tree maps have singular values in $[a,b]$, by the recoupling relation in H4, we get \[s_{\max}(U_B)\le\frac{b+r}{a},\qquad s_{\min}(U_B)\ge\frac{(a-r)_+}{b}. \]
The upper bound will control both squared deviations. For $0\le r\le a$,
\[
    \frac{(b+r)^2}{a^2}+\frac{(a-r)^2}{b^2}
    \ge\frac{b^2}{a^2}+\frac{a^2}{b^2}\ge2,
\]
For when $r\ge a$, the upper deviation already exceeds 1. Hence
\[
    \norm{U_B^\dagger U_B-I}\le\frac{(b+r)^2}{a^2}-1.
\]
The asymptotic statement will then follow from the zero-locus theorem.
\end{proof}

\begin{corollary}[]\label{charged:cor:cost}
Suppose $\chi_\ell=r_\ell=0$ and that $\C$ is non-unitarizable. Put
\[
    D=\Def(\C)>0,\qquad
    \theta_D=\frac{\sqrt{1+D}-1}{\sqrt{1+D}+1},\qquad
    c_{\C}=\frac{\theta_D}{m+\theta_D}>0.
\]
Then, at every scale satisfying the remaining hypotheses,
\begin{equation}\label{charged:eq:cost}
    q_\ell K_\ell^2\ge c_{\C}.
\end{equation}
In particular, if $q_\ell\le A e^{-\ell/\xi}$ for $A,\xi>0$, then
\begin{equation}\label{charged:eq:exp}
    K_\ell\ge\sqrt{\frac{c_{\C}}A}\,e^{\ell/(2\xi)}.
\end{equation}
\end{corollary}

\begin{proof}
Set $\eps=q_\ell K_\ell^2$. If $\eps<1/(m+1)$, then $d=m\eps/(1-\eps)<1$, and \eqref{charged:eq:main} gives
\[
    D\le\frac{4d}{(1-d)^2}.
\]
Let $s=\sqrt{1+D}$. $s\le(1+d)/(1-d)$, hence $d\ge\theta_D$. Substituting $d=m\eps/(1-\eps)$ gives $\eps\ge\theta_D/(m+\theta_D)$. The same inequality holds for $\eps\ge1/(m+1)$ since $0<\theta_D<1$. It is clear the exponential bound follows immediately.
\end{proof}

\paragraph{Distance from unitary recoupling.}
There is also a distance bound.
\begin{equation}\label{sn:polar-distance}
    \inf_{W\text{ unitary}}\norm{U-W}=\max_j|s_j-1|.
\end{equation}
Indeed, the polar unitary gives this distance, while for any right singular vector $v_j$ and unitary $W$, the reverse triangle inequality gives $\norm{(U-W)v_j}\ge|s_j-1|$.

Let $e=\max_j|s_j-1|$. We have $\norm{U^\dagger U-I}\le e(2+e)$. Thus every shared elementary metric will satisfy
\begin{equation}\label{sn:recoupling-distance}
    \max_B\inf_{W_B\text{ unitary}}\norm{U_B(g)-W_B}
    \ge\sqrt{1+D_{\mathrm{assoc}}(\mathcal C)}-1.
\end{equation}

The unitary matrices $W_B$ are chosen independently. If we require them to satisfy the pentagon relations, we will only increase the distance.

For $n$ stacked spherical inputs, the bulk bound implies that
\begin{equation}\label{sn:stack-distance}
    \max_B\inf_{W_B\text{ unitary}}\norm{U_B(g)-W_B}
    \ge r(\mathcal A)^n-1.
\end{equation}
For a single Yang--Lee bulk, the exact bound is $\lambda-1$, where $\lambda$ is defined in Section~\ref{sec:square}.

\subsection{Controlling the implementation on charged codes}\label{subsec:control}

\subsubsection{Splitting and merging retractions}

For the patch lifts that were constructed above, by controlling H2 means that we bound
\[
    \sup_{v\ne0,s}
    \frac{\norm{E_R(v)}}{\norm{(E_R(v)\otimes I_s)|_{\K_{c,s}}}}.
\]

For each $p=(a,b;c)$, fix a basis $e_1,\ldots,e_n$ of $V_{ab}^c$. Let $t_\alpha=t(e_\alpha)$. Suppose there are code-preserving maps $r_\alpha:\K_{ab,s}\to\K_{c,s}$ such that
\begin{equation}\label{charged:eq:retract}
    r_\alpha t_\beta=\delta_{\alpha\beta}I_{\K_{c,s}}.
\end{equation}
Assume that the physical splitting and merging maps satisfy the uniform bounds
\begin{equation}\label{charged:eq:columnbounds}
    \norm{[\widehat t_1\ \cdots\ \widehat t_n]}\le S_p,
    \qquad
    \norm{[\widehat r_1\ \cdots\ \widehat r_n]^{\mathsf T}}\le R_p.
\end{equation}
The second array is regarded as a column. Take the bounds uniformly over contexts and scales. Define $B=\max_p S_pR_p$.

\begin{theorem}\label{charged:thm:retract}
Replace H2 by \eqref{charged:eq:retract}--\eqref{charged:eq:columnbounds}, with the same H1, H3, and H4. The main defect bound will hold when
\begin{equation}\label{charged:eq:dretract}
    d_\ell'=m q_\ell B^2+2M\chi_\ell(1+q_\ell B^2)
\end{equation}
for whenever $d_\ell'<1$. If $\chi_\ell=r_\ell=0$ and $\C$ is non-unitarizable, then
\begin{equation}\label{charged:eq:retractfloor}
    q_\ell B^2\ge\frac{\theta_D}{m},
    \qquad \theta_D=\frac{\sqrt{1+D}-1}{\sqrt{1+D}+1}.
\end{equation}
Fixed $B$ and $q_\ell,\chi_\ell,r_\ell\to0$ will imply unitarizability.
\end{theorem}

\begin{proof}
For $v\in\mathbb C^n$ and $\psi\in\K_{c,s}$, the column $r=(r_\alpha)_\alpha$ satisfies $rt(v)\psi=v\otimes\psi$.
\[
    t(v)^\dagger t(v)\succeq R_p^{-2}\norm{v}^2 I.
\]
Take the requirements from H1 to get $g_p\succeq R_p^{-2}I$. Now, if $v_\alpha$ belongs to a $g_p$-orthonormal basis, $\norm{v_\alpha}\le R_p$.
\[
    \norm{\widehat t(v_\alpha)}\le S_p\norm{v_\alpha}\le S_pR_p\le B.
\]
Thus H1 bounds every normalized same-charge Gram entry by $q_\ell B^2$, and also gives $\norm{\widetilde t_\alpha}^2\le1+q_\ell B^2$. By applying the cross-charge estimate from \eqref{charged:eq:layer} we get \eqref{charged:eq:dretract}. Then, as before, the two-layer and singular-value arguments then follow.

For $\chi_\ell=r_\ell=0$, put $d'=mq_\ell B^2$. If $d'<1$, then $D\le4d'/(1-d')^2$ will imply that $d'\ge\theta_D$. When $d'\ge1$, \eqref{charged:eq:retractfloor} immediately follows.
\end{proof}

\begin{corollary}\label{charged:cor:fixedpatch}
Suppose fixed charged patch maps and virtual retractions---which realize all fusion vertices---are injective on the supports required for gluing and satisfy H1, H3, and H4. If $q_\ell,\chi_\ell,r_\ell\to0$, then $\C$ must be unitarizable.
\end{corollary}

The patch maps must be injective on the charged and punctured supports which are used in every gluing.

\subsubsection{Local charged density matrices}

Now suppose that the untransported splittings act on a finite region $R$. For each incoming charge $c$, choose a density matrix $\rho_{c,\ell}$ which has supports $S_{c,\ell}$ that contain the local supports of every relevant charged code.
\[
    (P_{S_{c,\ell}}\otimes I)P_{c,s}=P_{c,s}.
\]
For a finite family of contexts, one may take a positive mixture of their reduced and maximally mixed code states. The support is the span of their local supports. Use its expectation for the common scalar in H1.

Now, assume that the charged local estimate for every needed local product $O$ is
\begin{equation}\label{charged:eq:rho-code}
    \norm{P_{c,s}(O\otimes I)P_{c,s}-\Tr(\rho_{c,\ell}O)P_{c,s}}\le q_\ell\norm{O}.
\end{equation}
Let $\lambda_\ell$ be the minimum over finitely many incoming charge types of the smallest nonzero eigenvalue of $\rho_{c,\ell}$. We would need to further restrict each local splitting to their input support---replace $E(v)$ with $E(v)P_{S_{c,\ell}}$---the action on every charged code will still be unchanged.

\begin{theorem}\label{charged:thm:rho}
Under this version,
\begin{equation}\label{charged:eq:drho}
    d_\ell''=m\frac{q_\ell}{\lambda_\ell}+2M\chi_\ell\left(1+\frac{q_\ell}{\lambda_\ell}\right)
\end{equation}
when $d_\ell''<1$. If $\chi_\ell=r_\ell=0$ and $\C$ is non-unitarizable, then
\begin{equation}\label{charged:eq:rhofloor}
    \frac{q_\ell}{\lambda_\ell}\ge\frac{\theta_D}{m},
    \qquad \lambda_\ell\le\frac{m}{\theta_D}q_\ell.
\end{equation}
\end{theorem}

\begin{proof}
The elementary metric is $g_p(v,w)=\Tr(\rho_{c,\ell}E(v)^\dagger E(w))$. For reduced input support, positivity implies that
\[
    g_p(v,v)\ge\lambda_\ell\Tr(E(v)^\dagger E(v))
    \ge\lambda_\ell\norm{E(v)}^2.
\]
For a $g_p$-orthonormal vertex basis, each $E_\alpha$ must have a norm of at most $\lambda_\ell^{-1/2}$. ~\eqref{charged:eq:rho-code} bounds every same-charge Gram error by $q_\ell/\lambda_\ell$; its diagonal gives a bound of $1+q_\ell/\lambda_\ell$ on the squared physical splitting norm. The proof of Theorem~\ref{charged:thm:retract} then applies with $q_\ell B^2$ replaced by $q_\ell/\lambda_\ell$, which results in \eqref{charged:eq:drho} and \eqref{charged:eq:rhofloor}.
\end{proof}

Under these hypotheses, if $q_\ell\le A e^{-\ell/\xi}$, then
\[
    \lambda_\ell\le\frac{mA}{\theta_D}e^{-\ell/\xi}.
\]
Thus, as the charges are separated, some direction in the local charged support must receive vanishing weight.

If the patch $R$ is fixed and the density matrices $\rho_{c,\ell}$ converge without losing support rank, their smallest nonzero eigenvalues must remain bounded away from zero. In that case, $q_\ell,\chi_\ell,r_\ell\to0$ will imply that it is unitarizable.

\subsection{Consequences for a charged energy band}\label{subsec:chargedband}

We now consider what the obstruction implies for the energy spectrum. A local operator that acts nontrivially within a degenerate energy band will generally split the degeneracy under a small perturbation. To apply our bounds, however, we would first need the charged codes and local operators to be realized in the same Hamiltonian system.

Suppose that, for each incoming charge $c$, the required spectator code spaces embed isometrically into a common charged band $\mathcal P_{c,\ell}$ with orthogonal projector $P_{c,\ell}$. We expect for the splitting and merging operators to then act consistently on these subspaces. We will, however, also further require that they must be within the same Hamiltonian system.

For the finitely many fusion neighborhoods, consider
\begin{equation}\label{charged:eq:sensitivity}
    b_\ell=\max_{c,R}\ \sup_{\substack{O_R=O_R^\dagger\\\norm{O_R}\le1}}
    \operatorname{diam}\operatorname{spec}
    \left(P_{c,\ell}(O_R\otimes I)P_{c,\ell}|_{\mathcal P_{c,\ell}}\right),
\end{equation}
where the diameter is the difference between the largest and smallest eigenvalues.

Choose the H1 reference state on a represented reference code subspace where the retraction identities hold. The expectation for Hermitian $O_R$ lies between the extreme eigenvalues of the band compression, so the compression differs by at most $b_\ell\norm{O_R}$ from this scalar. As for a general $O_R=A+iB$, the two Hermitian parts have norm at most $\norm{O_R}$, and thus H1 holds with $q_\ell=2b_\ell$, also after restriction to the spectator subspaces.

\begin{corollary}\label{charged:cor:band}
Assume the common-band realization as suggested above, with bounded splitting and merging as in Theorem~\ref{charged:thm:retract}. If $\chi_\ell=r_\ell=0$ and $\C$ is non-unitarizable, then
\begin{equation}\label{charged:eq:bandfloor}
    b_\ell\ge\frac{\theta_D}{2mB^2}>0.
\end{equation}
In particular, there is a local Hermitian operator of norm one whose compression to a charged band has an order-one spectral spread, independent of separation when $B$ is fixed.
\end{corollary}
\begin{proof}
The result follows from \eqref{charged:eq:retractfloor} with $q_\ell=2b_\ell$.
\end{proof}

Suppose the charged band is an isolated degenerate eigenspace of a finite-dimensional Hermitian Hamiltonian $H_\ell$ that has energy $E_\ell$. Under a local perturbation $H_\ell+tO_R$, degenerate perturbation theory will give
\[
    E_j(t)=E_\ell+t\mu_j+O(t^2),\qquad
    \{\mu_j\}=\operatorname{spec}(P_{c,\ell}O_RP_{c,\ell}|_{\mathcal P_{c,\ell}}).
\]
Under the hypotheses of Corollary~\ref{charged:cor:band}, choosing its operator gives
\begin{equation}\label{charged:eq:firstorder}
    \lim_{t\downarrow0}\frac{\max_jE_j(t)-\min_jE_j(t)}{t}
    \ge\frac{\theta_D}{2mB^2}.
\end{equation}
One should note that showing a first-order splitting within the charged band does not show a closure of the isolation gap.

%% file: 05-measurements.tex
\section{Fusion measurements}\label{part:measurements}\label{sec:measurements}

Another way to look at the inner product is through fusion measurements without constructing splitting maps. One can, for example, measure the intermediate charge of a neighboring pair or a coherent superposition of different fusion channels and see what happens. In a physical Hilbert space, such outcomes will be described by positive effects. The corresponding categorical operators, however, are initially only idempotents. There is no reason for them to be necessarily compatible with the physical adjoint.

\subsection{The realization problem for fusion measurements}\label{subsec:measurementproblem}

First let us consider the case with three Yang--Lee charges with total charge $\tau$. Let $P$ project onto the vacuum fusion channel of the first pair. Similarly, let $Q$ project onto that of the last pair. If both were orthogonal projections on a Hilbert space, then they must satisfy $0\le PQP\le P$.

Let $\mathcal Y$ denote the Yang--Lee category and $\tau\otimes\tau\cong\one\oplus\tau$. Let $\phi=(1+\sqrt5)/2$. On $\operatorname{Hom}(\tau,\tau^{\otimes3})$, the associator in the symmetric gauge is
\begin{equation}\label{measure:eq:F}
    F=\begin{pmatrix}-\phi&i\sqrt\phi\\ i\sqrt\phi&\phi\end{pmatrix},
    \qquad F^{-1}=F.
\end{equation}
In the left fusion basis, the two selectors are
\begin{equation}\label{measure:eq:YL}
    P=\begin{pmatrix}1&0\\0&0\end{pmatrix},\qquad
    Q=FPF=
    \begin{pmatrix}
    \phi^2&-i\phi\sqrt\phi\\
    -i\phi\sqrt\phi&-\phi
    \end{pmatrix}.
\end{equation}
Although $Q^2=Q$, with direct multiplication, we show that
\begin{equation}\label{measure:eq:YLoverlap}
    Q^2=Q,\qquad PQP=\phi^2P,\qquad P(I-Q)P=-\phi P.
\end{equation}

Thus the two vacuum selectors cannot both be orthogonal projections. And since $PQP=\phi^2P$ is invariant under similarity transformations, there is no choice of positive metric which can remove this obstruction.

\subsection{Constructing a coherent measurement family}\label{sec:operational}

It is important to note that the vacuum selectors for a general fusion category might not detect an obstruction. Therefore, we should consider the measurements of coherent superpositions of fusion channels.

Set $X=\bigoplus_a a$ and
\[
    M_c=\operatorname{Hom}(c,X\otimes X)=\bigoplus_{a,b}V_{ab}^c,
    \qquad A_2=\operatorname{End}(X\otimes X)=\bigoplus_c\operatorname{End}(M_c).
\]
Let us choose normalized positive elementary forms $g$ on the spaces $V_{ab}^c$ choose for $E^c_{ij}$ to be the matrix units in orthonormal bases of $M_c$. The indices include external pairs and fusion multiplicities. The canonical $(\one,c)$ channel are labeled by $i=0$.

Consider the operators for each of the matrix blocks
\begin{align}
    D_i &= E_{ii}, \nonumber\\
    Q^{R,\pm}_{ij} &= \tfrac12(D_i+D_j\pm E_{ij}\pm E_{ji}), \label{op:family}\\
    Q^{I,\pm}_{ij} &= \tfrac12(D_i+D_j\pm i(E_{ij}-E_{ji})), \qquad i<j.\nonumber
\end{align}

Denote the resulting family by $\mathscr Q_g$ so it will contain $\sum_c(2r_c^2-r_c)$ projections where $r_c=\dim M_c$.

The real superpositions determine $E_{ij}+E_{ji}$; the imaginary ones determine $i(E_{ij}-E_{ji})$: both together will determine the off-diagonal matrix units $E_{ij}$ and $E_{ji}$.

There is also a reason to include different external label pairs: the two adjacent-pair placements must share the same physical Hilbert space and calibration. By composing matrix units between different pairs, we can in fact reduce either fusion tree to the same unit-channel state whose norm is what gives the common normalization for the two tree metrics.

We have so far chosen the elementary forms $g$ before defining the measurements, but they can alternatively be recovered from the physical adjoint. Suppose there is a faithful unital representation $\pi$ for $A_2$ whose image is closed under the Hilbert-space adjoint. Suppose that the selectors onto $V_{ab}^c\subset M_c$ can be represented by orthogonal projections and that each canonical matrix unit interchanging $(\one,c)$ and $(c,\one)$ will have a representation given by a partial isometry. The physical adjoint will then determine normalized positive elementary forms $g_{ab}^c$ in which all the operators in \eqref{op:family} will be physical orthogonal projections.

Indeed, on each simple block $\operatorname{End}(M_c)$, the physical adjoint is induced by a positive form on $M_c$, unique up to an overall positive scalar, which follows from the structure of finite-dimensional $C^*$-representations. Orthogonality of the external-pair selectors makes the form block diagonal across $V_{ab}^c$. The unit-channel partial isometry gives equal norms to the canonical $(\one,c)$ and $(c,\one)$ vectors, so we can rescale the form to make both of them to be unit vectors: restricting to each $V_{ab}^c$ results in the required elementary metrics and so the operators in \eqref{op:family} become orthogonal projections.

Let
\[
    H_3=\bigoplus_d\operatorname{Hom}(d,X^{\otimes3}).
\]
Now, choose one parenthesization as coordinates. Let $P_\alpha$ for the two placements $\ell(q)=q\otimes I_X$ and $r(q)=I_X\otimes q$ of every $q\in\mathscr Q_g$, with the right placement expressed using the exact associator. Let $A:H_3\to\mathcal H$ be any injective encoding and $0\le E_\alpha\le I$ arbitrary physical measurement effects. Then, define
\begin{equation}\label{op:error}
    \eta=\max_\alpha
    \|(E_\alpha A-AP_\alpha)(A^\dagger A)^{-1/2}\|.
\end{equation}
The map $A(A^\dagger A)^{-1/2}$ is an isometry, so $\eta$ measures the residual on a unit encoded input.

Figure~\ref{fig:measurement-anchor} shows the two collapses to the common unit channel.

\begin{figure}[tbp]
\centering
\begin{tikzpicture}[x=1cm,y=1cm,>=Stealth,font=\small,line width=.8pt]
\begin{scope}[shift={(-3.3,1.6)},scale=.6]
\draw (-1.5,1.8)--(-.7,.8)--(0,0)--(0,-.65);
\draw (0,1.8)--(-.7,.8);
\draw (1.5,1.8)--(0,0);
\node[above] at (-1.5,1.8) {$a$};\node[above] at (0,1.8) {$b$};\node[above] at (1.5,1.8) {$c$};
\node[left] at (-.5,.5) {$e$};\node[below] at (0,-.65) {$d$};
\end{scope}
\begin{scope}[shift={(3.3,1.6)},scale=.6]
\draw (1.5,1.8)--(.7,.8)--(0,0)--(0,-.65);
\draw (0,1.8)--(.7,.8);
\draw (-1.5,1.8)--(0,0);
\node[above] at (-1.5,1.8) {$a$};\node[above] at (0,1.8) {$b$};\node[above] at (1.5,1.8) {$c$};
\node[right] at (.5,.5) {$f$};\node[below] at (0,-.65) {$d$};
\end{scope}
\draw[<->] (-1.8,2.3)--node[above] {$F_B$} (1.8,2.3);
\node[draw,rounded corners=2pt,inner sep=7pt] (O) at (0,-1.3) {common anchor $(\one,\one,d)$};
\draw[->,diagramblue] (-3.3,.8)--node[below left,align=center] {$C_{L,i}$\\two matrix units} (O);
\draw[->,diagramblue] (3.3,.8)--node[below right,align=center] {$C_{R,j}$\\three matrix units} (O);
\node at (0,-2.2) {$C_{T,i}A_{T,j}=\delta_{ij}p_d$};
\end{tikzpicture}
\caption{Reconstructing a metric from coherent measurements..}
\label{fig:measurement-anchor}
\end{figure}

\subsection{The coherent-measurement obstruction}\label{subsec:measurementtheorem}

The off-diagonal matrix units allow us to collapse a fusion channel to a unit channel. Applying them successively to either fusion tree gives the same anchor $(\one,\one,d)$, as shown in Figure~\ref{fig:measurement-anchor}. As such, we can obtain a quantitative bound without assuming factorization of the physical metric.

\begin{theorem}
\label{op:main}
For every complex fusion category, set
\[
    D=D_{\mathrm{assoc}}(\mathcal C),\qquad
    N=\max_{a,b,c,d}\dim\operatorname{Hom}(d,a\otimes b\otimes c).
\]
Every local calibration $g$, injective encoding $A$, and family of physical effects in \eqref{op:error} obey
\begin{equation}\label{op:bound}
    12N\eta(1+2\eta)^5\ge\frac{D}{D+2}.
\end{equation}
If $\eta=0$, the physical metrics on both fusion trees agree with those induced by the elementary forms, up to a common positive scalar depending only on the total charge.
\end{theorem}

\begin{proof}
Let $G=A^\dagger A$ and $V=AG^{-1/2}$ such that $V$ is an isometry. The compressed effects are within $\eta$ of $\widetilde P_\alpha=G^{1/2}P_\alpha G^{-1/2}$. The off-diagonal matrix units are found with
\[
    E_{ij}=\tfrac12\bigl[(Q^{R,+}_{ij}-Q^{R,-}_{ij})-i(Q^{I,+}_{ij}-Q^{I,-}_{ij})\bigr].
\]
Now since the difference of two effects is a Hermitian contraction, each placed matrix unit $b$ will be within $2\eta$ of a contraction. For $K=1+2\eta$, applying the same construction to its formal reverse $b^*$ gives,
\begin{equation}\label{op:unit}
    \|\widetilde b\|\le K,\qquad
    \|\widetilde b^\dagger-\widetilde{b^*}\|\le4\eta.
\end{equation}
The diagonal units satisfy the same bounds.

Both fusion trees can be mapped to a common unit channel. Let $L_{ab;e,\mu}$ be the matrix unit taking the chosen $(a,b)\to e$ channel to $(\one,e)$. Let $T_a$ take $(a,\one)$ to $(\one,a)$. For the left and right trees respectively
\begin{align*}
    C_{L,i}&=r(L_{ec;d,\nu})\ell(L_{ab;e,\mu}),\\
    C_{R,j}&=r(L_{af;d,\delta})\ell(T_a)r(L_{bc;f,\gamma}).
\end{align*}

The left collapse simply uses two matrix units; the right will also require an additional move $T_a$: yet, both end at the same anchor $(\one,\one,d)$. Let $A_{T,i}$ be the formal reverse of $C_{T,i}$, which we obtain when we take the anchor vector $w_d$ to the corresponding tree vector in a way such that
\begin{equation}\label{op:collapse}
    C_{T,i}A_{T,j}=\delta_{ij}p_d,
\end{equation}
where $p_d$ is the anchor selector.

When $\eta=0$, formal reversal is the physical adjoint, so
\[
    \langle A_{T,i}w_d,A_{T,j}w_d\rangle_G
    =\delta_{ij}\langle w_d,w_d\rangle_G.
\]
Both tree bases have Gram matrix $h_dI$ in elementary orthonormal coordinates. $h_d=\langle w_d,w_d\rangle_G$: in the original coordinates, the tree metrics are $h_dG_{T,B}$.

A product of $k$ matrix units satisfy
\[
    \|\widetilde A_{T,i}\|\le K^k,\qquad
    \|\widetilde A_{T,i}^\dagger-\widetilde C_{T,i}\|\le4k\eta K^{k-1}.
\]
Since the left and right products have lengths two and three respectively, \eqref{op:collapse} bounds each Gram-entry error by $12\eta K^5$.

Now let $W_L$ and $W_R$ be the corresponding tree-column maps. Each associator block has dimension at most $N$, so
\[
    \|W_L^\dagger W_L-I\|,\ \|W_R^\dagger W_R-I\|
    \le\varepsilon:=12N\eta K^5.
\]
 If $\varepsilon<1$,
\[
    \frac{1-\varepsilon}{1+\varepsilon}I
    \preceq U_B(g)^\dagger U_B(g)
    \preceq\frac{1+\varepsilon}{1-\varepsilon}I.
\]
Thus $D\le2\varepsilon/(1-\varepsilon)$---equivalently $\varepsilon\ge D/(D+2)$. The latter inequality is trivial for $\varepsilon\ge1$. Finally, the lower bound follows from $K^5\le32$ when $\eta\le1/2$, trivially so when $\eta>1/2$.
\end{proof}

\subsection{What the obstruction means for an implementation}\label{subsec:measurementconsequences}
Now, suppose $D>0$. If all the measurements had zero error, their physical adjoints gives coherent positive forms on both fusion trees, which means that every associator would then be unitary. That contradicts the assumption that $D>0$.

One might try to enlarge the physical Hilbert space or allow the effects to take states outside the encoded subspace, but because their compressions are still positive effects, the bound will still remain. A different issue arises when the residual is measured without normalizing the encoding. In that case, a small error may simply come from a poorly conditioned $A$.

Let $\beta_N(D)$ be the unique nonnegative solution of
\[
    12N\beta(1+2\beta)^5=\frac{D}{D+2}.
\]
For $D>0$, we have $\beta_N(D)>0$. Now suppose that the error is instead reported with
\[
    e_A=\max_\alpha\frac{\|E_\alpha A-AP_\alpha\|}{\|A\|},
    \qquad \kappa(A)=\frac{\|A\|}{s_{\min}(A)}.
\]
From \eqref{op:error}, we have $\eta\le\kappa(A)e_A$, and therefore
\begin{equation}\label{op:condition}
    \kappa(A)e_A\ge\beta_N(D).
\end{equation}
For example, if $e_{A,\ell}\le Ce^{-\ell/\xi}$, then $\kappa(A_\ell)\ge\beta_N(D)e^{\ell/\xi}/C$. $\kappa(A)$ concerns the encoding itself, while $K_\ell$ in Sections~\ref{subsec:localmaps}--\ref{subsec:control} is what actually bounds the physical operator norm relative to its action on the code.

Let $\Delta_{\mathrm{obs}}(\mathcal C)$ be the infimum of \eqref{op:error} over all choices of $g$, $A$. We can bound this quantity from both sides using the associator defect, such that the constants depending only on the fusion rules. As in the inverse-norm estimate, set
\[
    w_B=d_ad_bd_cd_d,\qquad
    \Phi=\sum_a d_a^2,\qquad
    J=\max_B(\Phi^3/w_B-1).
\]

\begin{proposition}\label{op:comparison}
For every complex fusion category,
\begin{equation}\label{op:twosided}
    \beta_N(D)\le\Delta_{\mathrm{obs}}(\mathcal C)
    \le\tfrac12\bigl((1+D)^{(J+1)/2}-1\bigr).
\end{equation}
If $D<1$, the upper bound can be replaced by
\begin{equation}\label{op:small}
    \tfrac12\left(\sqrt{\frac{1+D}{1-D}}-1\right).
\end{equation}
Consequently, $\Delta_{\mathrm{obs}}=0$ if and only if $\mathcal C$ is unitarizable. In that case, a zero-error realization will exist.
\end{proposition}

\begin{proof}
The lower bound follows immediately from Theorem~\ref{op:main}. For the upper bound, choose elementary metrics attaining $D$. Equip $H_3$ with the corresponding left-tree Hilbert metric. The left-placed observables must then be orthogonal projections.

Now, let $U$ be the orthogonal direct sum of the recoupling blocks in elementary orthonormal coordinates. A right-placed observable has the form $U^{-1}pU$, where $p$ is an orthogonal projection in the right-tree coordinates. By polar decomposition, $U=WP$. Let $e=W^\dagger pW$. Then
\[
    U^{-1}pU=P^{-1}eP,\qquad
    \|P^{-1}eP-e\|
    \le\frac{s_{\max}(P)-s_{\min}(P)}{2s_{\min}(P)}.
\]
To prove the second inequality, notice that the commutator $[e,P]$ has only off-diagonal blocks relative to $e$. Subtracting the midpoint of the spectrum of $P$ will bound its norm by half the spectral diameter, and since $e$ is an orthogonal projection, we can use it as the physical effect.

By determinant balance, $\|U\|\le\sqrt{1+D}$ and $\|U^{-1}\|\le(1+D)^{J/2}$, which gives \eqref{op:twosided}. For $D<1$, the Gram bound will also give $\|U^{-1}\|\le(1-D)^{-1/2}$, which proves \eqref{op:small}.

For $D=0$, the normalized recouplings are unitary and both placements become orthogonal projections. The converse will then follow naturally from the lower bound and Theorem~\ref{fusion:zero}.

\end{proof}

Near zero, the bounds will give that $\beta_N(D)=D/(24N)+O(D^2)$ and \eqref{op:small} gives $D/2+O(D^2)$. Thus $\Delta_{\mathrm{obs}}(\mathcal C)$ is bounded above and below linearly in $D$ near zero---the constants depends on the fusion rules.

The effects which are constructed here act on the full encoded space. We would require more locality assumptions in order to find the realization by operators supported on specified lattice patches.

\subsection{Two charge selectors and sharp positivity bounds}\label{subsec:selectors}

Lets return to the case of two selectors.  For orthogonal projections $P$ and $Q$, the spectrum of $PQP$ on $\ran P$ lies in $[0,1]$. If an algebraic eigenvalue lies outside this interval, no choice of positive inner product can make both selectors orthogonal projections. But how closely can such a pair be realized by physical measurements?

Let $F:V_L\to V_R$ be an exact associator between fusion-tree spaces, and let $D_i^L,D_j^R$ select the corresponding intermediate-charge summands. On $V_L$, define the categorical idempotents
\begin{equation}\label{measure:eq:PQ}
    P_i=D_i^L,\qquad Q_j=F^{-1}D_j^R F.
\end{equation}
An injective encoding $A$ of $V_L$ will give us a positive metric $G=A^\dagger A$.

Figure~\ref{fig:selector-overlap} compares the projection geometry with the examples of Section~\ref{subsec:YLmeasure}.

\begin{figure}[tbp]
\centering
\begin{tikzpicture}[x=1cm,y=1cm,>=Stealth,font=\small,line width=.8pt]
\begin{scope}[shift={(-3.5,0)}]
\draw[->] (-.5,0)--(3.4,0) node[below] {$\ran P$};
\draw[->,diagramblue] (-.3,-.258)--(2.5,2.15) node[above] {$\ran Q$};
\draw (.7,0) arc(0:41:.7);\node at (1,.35) {$\theta$};
\draw[->] (0,0)--(2.7,0);
\draw[->,dashed] (2.7,0)--(1.55,1.33);
\draw[->,dashed] (1.55,1.33)--(1.55,0);
\node[align=center] at (1.3,-.9) {orthogonal rank-one projections\\$PQP=\cos^2\theta\,P$};
\end{scope}
\begin{scope}[shift={(1.6,0)}]
\draw[->] (0,.5)--(4.8,.5) node[right] {$t$};
\draw[line width=3pt,diagramblue!50] (.2,.5)--(2,.5);
\foreach \x/\l in {.2/0,2/1} {\draw (\x,.4)--(\x,.6);\node[below] at (\x,.35) {$\l$};}
\fill[diagramblue] (.9,.5) circle(2pt) node[above=5pt] {$\phi^{-2}$};
\fill (3,.5) circle(2pt) node[above=5pt] {$\phi^2$};
\fill (4.3,.5) circle(2pt) node[above=5pt] {$\phi^4$};
\node[align=center] at (2.3,-.9) {allowed overlap $[0,1]$\\Yang--Lee lies outside};
\end{scope}
\end{tikzpicture}
\caption{The overlap test is a geometric obstruction to positivity.}
\label{fig:selector-overlap}
\end{figure}

If both $P_i$ and $Q_j$ are orthogonal projections in the same positive metric, then $0\le P_iQ_jP_i\le P_i$. Thus
\[
    \spec\!\left(P_iQ_jP_i\big|_{\ran P_i}\right)\subset[0,1].
\]
For one-dimensional summands, simply direct multiplication gives the eigenvalue
\begin{equation}\label{measure:eq:entry}
    t_{ij}=(F^{-1})_{ij}F_{ji}.
\end{equation}
For higher multiplicities, we will instead take the spectrum of the block $(F^{-1})_{ij}F_{ji}$. Now, since the spectrum is unchanged by similarity, the condition does not depend on the choice of bases or faithful encoding.

Of course this only tests one pair of selectors. To unitarize the category, the same elementary forms must work for every associator block.

We do allow arbitrary physical effects $0\le E\le I$. Each belongs to a binary measurement $\{E,I-E\}$ \cite[Section 2.3.1, Definition 2.34]{Watrous}.

\begin{theorem}\label{measure:thm:quant}
Let $P,Q$ be nonzero idempotents on a finite-dimensional vector space. Suppose $PQ$ has an eigenvalue $t>1$. For every positive metric $G$, set
\[
    p=G^{1/2}PG^{-1/2},\qquad q=G^{1/2}QG^{-1/2}.
\]
For arbitrary effects $E_1,E_2$ on this Hilbert space,
\begin{align}
    \max\{\norm p,\norm q\}&\ge\sqrt t,\label{measure:eq:norm}\\
    \max\{\norm{p-p^\dagger},\norm{q-q^\dagger}\}&\ge\sqrt{t-1},\label{measure:eq:dagger}\\
    \max\{\norm{p-E_1},\norm{q-E_2}\}&\ge
    h(t):=\frac{\sqrt t+\sqrt{t-1}-1}{2}.\label{measure:eq:effect}
\end{align}
For rank-one $P,Q$ with $\Tr(PQ)=t>1$, each bound is sharp over positive metrics. The last bound can be found using orthogonal projections as effects.
\end{theorem}

\begin{proof}
Since $t\le\norm{pq}\le\norm p\norm q$, there must be at least one idempotent has norm at least $\sqrt t$. Suppose it is $p$, for $\norm p=n\ge\sqrt t$. $p$ has the form
\[
    p=\begin{pmatrix}I&T\\0&0\end{pmatrix}.
\]
Let $a=\norm T$. Spectral estimates gives
\begin{equation}\label{measure:eq:idemnorm}
    n=\sqrt{1+a^2},\qquad
    \norm{p-p^\dagger}=a,\qquad
    \norm{2p-I}=n+\sqrt{n^2-1}.
\end{equation}

For any effect $E$, we have $\norm{2E-I}\le1$.
\[
    2\norm{p-E}\ge\norm{2p-I}-\norm{2E-I}\ge n+\sqrt{n^2-1}-1.
\]
Since $n\ge\sqrt t$, this proves all three lower bounds.

Now, for sharpness, set $n=\sqrt t$ and $a=\sqrt{t-1}$. Consider
\begin{equation}\label{measure:eq:balanced}
    p_*=\frac12\begin{pmatrix}n+1&a\\-a&1-n\end{pmatrix},\qquad
    q_*=\frac12\begin{pmatrix}n+1&-a\\a&1-n\end{pmatrix}.
\end{equation}
For this pair, $\Tr(p_*q_*)=t$. Taking $E_1=E_2=\operatorname{diag}(1,0)$ gives
\[
    \norm{p_*-E_1}=\norm{q_*-E_2}=\frac{n+a-1}{2}.
\]
The other two bounds can be found with \eqref{measure:eq:idemnorm}.

Now consider any rank-one pair $P,Q$ with $\Tr(PQ)=t>1$. On the span of their ranges, we can choose coordinates such that
\[
    P=\begin{pmatrix}1&0\\0&0\end{pmatrix},\qquad
    Q=\begin{pmatrix}t&b\\c&1-t\end{pmatrix},
    \qquad bc=t(1-t)\ne0.
\]
Every rank-one pair with $\Tr(PQ)=t>1$ up to common zero blocks is simultaneously similar to \eqref{measure:eq:balanced}. By polar decomposition, the similarity results in the positive metric for which the bounds are attained.
\end{proof}

\subsection{Yang--Lee measurements and local protection}\label{subsec:YLmeasure}
Let us go back to the original example of the Yang--Lee case and apply our new measurement to \eqref{measure:eq:YL}. Take $t=\phi^2$.

By Theorem~\ref{measure:thm:quant}, there is no positive metric can make both vacuum selectors orthogonal projections, so at least one must have an adjoint defect at least $\sqrt\phi$ and at least one must be at distance
\begin{equation}\label{measure:eq:YLconstant}
    \frac{\phi+\sqrt\phi-1}{2}
\end{equation}
or greater from any physical effect.

We can also derive this directly from the metric equations: if $G=G^\dagger>0$ and $GP=P^\dagger G$, then $G=\operatorname{diag}(u,v)$ with $u,v>0$ and by requiring $GQ=Q^\dagger G$, $u=-v$. There seems to be an indefinite solution but not a positive one.

Now suppose that $A:\mathbb C^2\to\mathcal H$ is nonzero and that $E_1A=AP$ and $E_2A=AQ$ for some physical effects $E_1$ and $E_2$. The kernel of $A$ must be invariant under both $P$ and $Q$. The only proper invariant lines of $P$ are the coordinate axes, neither of which is invariant under $Q$. Hence $A$ must be injective. The positive metric $G=A^\dagger A$ would then make both selectors orthogonal, which contradicts the calculation above. Thus no such nonzero encoding exists.

As another example, the Fibonacci has the same fusion rule, but its orthogonal pair selectors have overlap $\phi^{-2}\in(0,1)$. Galois conjugation changes this to $\phi^2>1$ in Yang--Lee, although the two categories have the same fusion rules.

Since $\mathcal Y$ is modular, the center is
\[
    Z(\mathcal Y)\simeq\mathcal Y\boxtimes\mathcal Y^{\rm rev}.
\]
This gives us the doubled sector structure of the nonunitary string net \cite[Eqs. (8)--(9)]{Lootens}---the center factorization follows from \cite[Proposition 8.20.12]{EGNO}. Here $\mathcal Y^{\rm rev}$ differs only in its braiding, but not its associator.

For the chiral particle $b=\tau\boxtimes\one$, the preceding two-dimensional calculation applies. Now, consider the diagonal particle $a=\tau\boxtimes\tau$. The space $\operatorname{Hom}(a,a^{\otimes3})$ has dimension four, with associator $F\otimes F$. Its two pair-vacuum selectors are
\[
    P_a=P\otimes P,\qquad Q_a=Q\otimes Q,
    \qquad P_aQ_aP_a=\phi^4P_a.
\]

Because both selectors have rank one, Theorem~\ref{measure:thm:quant} gives sharp bounds $\phi^2$, $\sqrt{\phi^4-1}$ and
\[
    h(\phi^4)=\frac{\phi^2+\sqrt{\phi^4-1}-1}{2}.
\]

Let $A:V\to\mathcal H$ be an injective encoding of either fusion space. Have
\[
    G=A^\dagger A,\qquad V_A=AG^{-1/2},\qquad V_A^\dagger V_A=I.
\]
Given two physical effects $E_1,E_2$, define
\begin{equation}\label{measure:eq:patcherror}
    \eta_A=\max\left\{
    \norm{(E_1A-AP_*)G^{-1/2}},\;
    \norm{(E_2A-AQ_*)G^{-1/2}}\right\},
\end{equation}
for the corresponding categorical pair $(P_*,Q_*)$,

Theorem~\ref{measure:thm:quant} will mean that
\begin{equation}\label{measure:eq:patchbound}
    \eta_A\ge h(\phi^2)\quad\text{or}\quad
    \eta_A\ge h(\phi^4),
\end{equation}
respectively. By normalizing by $G^{-1/2}$, the error will be independent of an overall rescaling of $A$.

Unlike the two-dimensional Yang--Lee pair the doubled vacuum selectors have a common kernel, which a noninjective encoding could discard. However, the full family of left and right channel selectors acts irreducibly, since every entry of $F\otimes F$ and its inverse is nonzero. Thus, any nonzero encoding that has all these selectors must be injective.

For the local splitting estimates, consider the bulk category $Z(\mathrm{YL})$---they are monoidally equivalent to $\mathrm{YL}^{\boxtimes2}$. It is multiplicity-free and the defect satisfies $\sqrt{1+D}=\lambda$. The local splitting bound becomes
\begin{equation}\label{sn:YLcost}
    q_\ell K_\ell^2\ge\frac{\lambda-1}{2\lambda}.
\end{equation}
If bounded splitting and merging retractions are available:
\begin{equation}\label{sn:YLretract}
    q_\ell B^2\ge\frac{\lambda-1}{\lambda+1}.
\end{equation}
If the charged codes lie in a common isolated degenerate band of the same Hamiltonian:
\begin{equation}\label{sn:YLband}
    b_\ell\ge\frac{\lambda-1}{2(\lambda+1)B^2}.
\end{equation}

\subsection{Discussion}\label{sec:conclusion}\label{physical:end}

For a PEPS, we can start with the charged and punctured supports of a finite patch. As in Section~\ref{subsec:lifts}, pseudoinverses give splitting and merging operators on these supports. The difficulty is whether these operators remain consistent under gluing and in the presence of different spectator charges. Of course bounded retractions or lower bounds on the nonzero local density eigenvalues would give the required norm control.

\begin{question}\label{question:patch}
It would be interesting to know: can these charged splittings and merges be realized on a patch whose size stays bounded as the spectators separate?
\end{question}

\begin{question}\label{question:measurement-locality}
When can the coherent measurements of Section~\ref{sec:operational} be implemented locally? What would the physical reason be for the same calibration to work in both placements?
\end{question}

It is worth noting that there are other approaches to obtaining a unitary charge category from a physical system. Under approximate Haag duality and additional sector assumptions, Ogata reconstructs a braided $C^*$-tensor category \cite[Setting 5.1 and Theorem 5.2]{Ogata}. To apply this result, we would need to identify the proposed center with the reconstructed physical category. Boundary approximate factorization might also imply a gap for a PEPS parent Hamiltonian, although this is no guarantee \cite[Theorems 8 and 10]{KLP}. The Galois-conjugated PEPS of \cite[discussion of Eqs. (6)--(9)]{Lootens} will retain their virtual fusion identities, but their compatibility with physical adjoints is still a tad unclear.

Finally, it should be worth noting that for a nonunitarizable category, the errors cannot tend to zero while all the required controls remain bounded. Our estimates, however, cannot tell us which assumption fails in a particular lattice model, which could be particularly interesting to know. Unfortunately, the theory presented does not give us that much information.

%% file: 06-appendix-fusion-bases.tex
\section{Fusion bases and normalized gauges}\label{app:gauges}

Under a passive change of elementary fusion bases $T_v$, the associator and Gram matrices transform as
\[
    F'_B=T_{R,B}^{-1}F_BT_{L,B},\qquad
    G'_{L,B}=T_{L,B}^\dagger G_{L,B}T_{L,B},\qquad
    G'_{R,B}=T_{R,B}^\dagger G_{R,B}T_{R,B}.
\]
Consider the matrices
\[
    Q_L=G_L^{1/2}T_L(G'_L)^{-1/2},\qquad
    Q_R=G_R^{1/2}T_R(G'_R)^{-1/2}.
\]
Both are unitary. Furthermore,
\[
    U'_B=Q_R^\dagger U_BQ_L.
\]
The singular values of $U_B$ (and $\Def$) are therefore independent of the fusion bases.

Now, for an active gauge transformation $F_B\mapsto S_RF_BS_L^{-1}$, the elementary polar decompositions $S_v=V_vP_v$ give
\begin{equation}\label{fusion:polar}
    (S\cdot F)_B=V_R(P\cdot F)_BV_L^\dagger.
\end{equation}

Suppose $S$ relates two normalized associators. To see normalization, let
\[
    u_a=S_{\one a}^{a},\qquad v_a=S_{a\one}^{a}.
\]

The block $(\one,a,b;c)$ gives $u_c=u_a$ whenever $N_{ab}^c>0$. Taking $b=a^*$ and $c=\one$ yields $u_a=u_{\one}$. Similarly, $v_a=v_{\one}$ and $u_{\one}=v_{\one}$ as they refer to the same unit vertex.

Rescaling every $S_v$ by the common scalar leaves the associator unchanged and normalizes the unit vertices. The positive polar factors then give the admissible elementary metrics in \eqref{direct:eq:defect}.

%% file: 07-appendix-coupled-descent.tex
\section{Coupled descent and stacking estimates}\label{app:gain}

The step bound depends on primarily the size of the gauge direction and the singular-value gap of the active blocks, so it should be easy to account for.

Suppose that each active block has a simple top singular value and that $\nu>0$. It has a derivative of at most $-\nu$ along a coherent direction $H$. Let
\begin{align*}
    h&=\max_B(\|H_{R,B}\|+\|H_{L,B}\|),\\
    \gamma&=\min_{B\ \text{active}}(1-s_2(F_B)/s_1(F_B)),\\
    \Delta&=a_0-\max_{B\ \text{inactive}}\log\|F_B\|.
\end{align*}
For scalar blocks, we take $s_2=0$. If every block is active, the restriction involving $\Delta$ below is unnecessary and can be safely omitted.

\begin{proposition}\label{coupled:gain}
\begin{equation}\label{coupled:step}
    0<t\le\min\left\{\frac{\gamma}{8h},\frac{\nu\gamma}{8h^2},\frac{\Delta}{2h}\right\},
\end{equation}
\begin{equation}\label{coupled:gainbound}
    \max_B\|F_B(t)\|\le e^{a_0-\nu t/2}.
\end{equation}
\end{proposition}

\begin{proof}
For an active block $B$, let $A(t)=e^{-a_0}F_B(t)$ so $\|A(0)\|=1$. Consider now the Hermitian dilation
\[
    A=\begin{pmatrix}0&F\\F^\dagger&0\end{pmatrix}.
\]
The top eigenvalue is one and has a gap of at least $\gamma$. If $\|E\|\le\gamma/4$, decomposition along the top eigenvector $w$ and its orthogonal complement gives
\[
    \lambda_{\max}(A+E)\le1+\langle w,Ew\rangle+\frac{2\|E\|^2}{\gamma}.
\]
Indeed, the complementary diagonal block is bounded above by $1-\gamma+\|E\|$, while the first diagonal entry is at least $1-\|E\|$. The off-diagonal norm is at most $\|E\|$, giving the bound by comparison with a $2\times2$ matrix.

The exponential series shows $\|F(t)-F\|\le e^{ht}-1$, with a first-order remainder bounded by $h^2t^2e^{ht}/2$. As for $ht\le\gamma/8$, we obtain
\[
    \log\|F(t)\|\le\log\|F\|-\nu t+\frac{4h^2t^2}{\gamma}\le a_0-\frac{\nu t}{2}.
\]
For inactive blocks,
\[
    \log\|F_B(t)\|\le a_0-\Delta+ht.
\]

Since $\nu\le h$, the condition $t\le\Delta/(2h)$ gives the same bound for inactive blocks.
\end{proof}

\paragraph{Subadditivity under stacking.}\label{app:subadditivity}
Let $b_n=\log(1+D_{\mathrm{bulk}}(\mathcal C^{\boxtimes n}))$ with $b_0=0$. By subadditivity, $b_{n+k}\le b_n+b_k$. For fixed $k\ge1$, let $n=qk+s$, where $0\le s<k$.
\[
    \frac{b_n}{n}\le\frac{qb_k}{n}+\frac{\max_{0\le s<k}b_s}{n}.
\]
Taking the upper limit as $n\to\infty$ gives $\limsup b_n/n\le b_k/k$. On the other hand, $\liminf b_n/n\ge\inf_{k\ge1}b_k/k$. Thus
\[
    \lim_{n\to\infty}\frac{b_n}{n}=\inf_{k\ge1}\frac{b_k}{k}.
\]
Finally, exponentiating gives the limit in \eqref{bulk:eq:rate}.

%% file: 08-appendix-yang-lee.tex
\section{The doubled Yang--Lee optimality}\label{app:double}

Here we will verify the doubled Yang--Lee minimum using the parameters and metric in \eqref{square:eq:lambda} and \eqref{square:eq:metrics}.

The nontrivial gauges are $S_{pa}^{a}=S_{ra}^{a}=k^{-1}$ and $S_{ap}^{a}=S_{ar}^{a}=k$. For $B_0=(a,a,a;a)$. Order the intermediate channels as $1,p,r,a$. The transformed matrix is
\begin{equation}\label{square:eq:fourblock}
    F_0(k)=D_k(F\otimes F)D_k,\qquad D_k=\operatorname{diag}(1,k,k,1).
\end{equation}
The two factors of $D_k$ come from the left and right paths through $p,r$: the former contribute the inverse of $k^{-1}$, while the latter contribute $k$. Keep in mind that this is not a diagonal similarity.

For $B_p=(a,p,a;p)$ and $B_r=(a,r,a;r)$, with intermediate channels $r,a$ and $p,a$ respectively. The transformed matrices are both
\begin{equation}\label{square:eq:twoblock}
    F_1(k)=\begin{pmatrix}1&0\\0&k^{-1}\end{pmatrix}F\begin{pmatrix}1&0\\0&k^{-1}\end{pmatrix}.
\end{equation}

The matrix $QFQ$---$Q=\operatorname{diag}(i,-1)$---is positive symmetric. Under $Q\otimes Q$, the four-channel block will decompose into the symmetric and antisymmetric parts and have eigenvalues given by those of $C(k)$, together with $1$ and $k^2$.
The reduced matrices are
\begin{align}
    C(k)&=\begin{pmatrix}\alpha&\beta k\\\beta k&\alpha k^2\end{pmatrix},\qquad \alpha=\phi^3,\quad\beta=2\phi\sqrt\phi,\quad\alpha^2-\beta^2=1,\label{square:eq:Cmatrix}\\
    B(k)&=\begin{pmatrix}\phi&\sqrt\phi/k\\\sqrt\phi/k&\phi/k^2\end{pmatrix}.\label{square:eq:Bmatrix}
\end{align}
Let $z=k^2$. A common eigenvalue $\lambda$ satisfies
\begin{equation}\label{square:eq:char}
    \lambda^2-\alpha(1+z)\lambda+z=0,\qquad z\lambda^2-\phi(1+z)\lambda+1=0.
\end{equation}
Eliminating $z$ gives us $(\lambda^2-1)(\lambda^2-(3\phi+1)\lambda+1)=0$. For $0<z<1$ and $\lambda>\alpha>\phi$, the solution is given in \eqref{square:eq:lambda}.

For optimality, first let $H_v$ be arbitrary Hermitian generators that vanish on the unit channels. Derivatives are given by \eqref{direct:eq:forwardactive}.

Each active block is diagonally unitarily congruent to a positive symmetric matrix at the metric, and thus its top left and right singular vectors have equal coordinate probabilities. Let $c>0$ be the common probability of channels $p$ and $r$ in $B_0$, and let $h>0$ be the common probability of channel $a$ in $B_p$ and $B_r$.

Define
\begin{align*}
    X(H)&=H_{ap}^{a}+H_{ar}^{a}-H_{pa}^{a}-H_{ra}^{a},\\
    Y(H)&=H_{pa}^{r}+H_{ar}^{p}-H_{ap}^{r}-H_{ra}^{p}.
\end{align*}
From \eqref{direct:eq:forwardactive}, it is clear that
\begin{align}
    L_0(H)&=cX(H),\notag\\
    L_p(H)&=(1-h)Y(H)+h(H_{pa}^{a}-H_{ap}^{a}),\label{square:eq:functionals}\\
    L_r(H)&=-(1-h)Y(H)+h(H_{ra}^{a}-H_{ar}^{a}).\notag
\end{align}
We can cancel whenever channel probabilities, here, terms involving $H_{aa}^{e}$, agree on the two trees.
Set
\begin{equation}\label{square:eq:weights}
    \theta_0=\frac{h}{h+2c},\qquad \theta_p=\theta_r=\frac{c}{h+2c}.
\end{equation}
These weights satisfy
\begin{equation}\label{square:eq:balance}
    \theta_0L_0(H)+\theta_pL_p(H)+\theta_rL_r(H)=0
    \qquad\text{for every elementary }H.
\end{equation}

Join the metric to any coherent metric $g$ by a product matrix geodesic. Let $f(t)$ be the weighted sum of the logarithmic norms of $B_0,B_p,B_r$ that have weights \eqref{square:eq:weights}. By convexity and \eqref{square:eq:balance},
\[f(1)\ge f(0)+f'(0)=\log\lambda.\]
Hence
\begin{equation}\label{square:eq:lowernorm}
    \max\{\|U_{B_0}(g)\|,\|U_{B_p}(g)\|,\|U_{B_r}(g)\|\}\ge\lambda
    \qquad\text{for every }g,
\end{equation}
and therefore $D_{\mathrm{assoc}}\ge\lambda^2-1$.

For the remaining 141 blocks, a Hilbert--Schmidt estimate will suffice. The scalar squared norms are among $z^{-2},z^{-1},1,z,z^2$, but the nonscalar squared Hilbert--Schmidt norms are among
\[
    \phi(z^{-2}+z)+\phi^2(z^{-1}+1),\quad
    (\phi+\phi^2)(1+z),\quad
    2\phi+\phi^2(z^{-2}+z^2),\quad
    2\phi+2\phi^2.
\]
obtained by summing the squared entry moduli $\phi^jz^{\varepsilon_{bc}^{f}+\varepsilon_{af}^{d}-\varepsilon_{ab}^{e}-\varepsilon_{ec}^{d}}$, given that $S_v=k^{\varepsilon_v}$ and $j$ is the sum of the one-factor entry exponents.

The radicals give
\begin{equation}\label{square:eq:rationalintervals}
    \frac85<\phi<\frac53,\qquad \frac{11}2<\lambda<6,\qquad \frac13<z<\frac38.
\end{equation}
The scalar expressions are bounded by $9$, and the four nonscalar expressions by $27,7,29,9$ respectively. Every inactive squared norm is therefore below $29<121/4<\lambda^2$.

The three active blocks have norm $\lambda$; every inactive block has strictly smaller norm. Since all blocks are invertible and $\lambda^2-1>1$, the candidate has a defect of $\lambda^2-1$. Together with \eqref{square:eq:lowernorm}, this proves Theorem~\ref{square:statement:1}.
\hfill$\square$

%% file: 09-appendix-dimensions.tex
\section{Dimension calculations for the quantum-group families}\label{app:dimensions}

\paragraph{The smallest squared norm in a modular category.}\label{app:modular-dimensions}
Let $\mathcal M$ be a modular category over $\mathbb C$, and write $q_j=\dim(j)$, $\Delta=\dim(\mathcal M)$, and $\Phi=\operatorname{FPdim}(\mathcal M)$. By \cite[proof of Proposition 9.4.2]{EGNO}, the Frobenius--Perron character is an $S$-matrix column character. Thus, for some simple object $r$,
\[
    \operatorname{FPdim}(j)=\frac{S_{jr}}{q_r}.
\]
Using orthogonality and the convention $S_{1r}=q_r$, we have
\[
    \Phi=\sum_j\operatorname{FPdim}(j)\operatorname{FPdim}(\bar j)
    =\frac{1}{q_r^2}\sum_j S_{jr}S_{\bar j r}
    =\frac{\Delta}{q_r^2}.
\]
Hence
\begin{equation}\label{bulk:eq:small-simple}
    |r|^2=q_r^2=\frac{\Delta}{\Phi}.
\end{equation}
Indeed, for any column character $\chi_j(i)=S_{ij}/q_j$, the value $\chi_j(i)$ is an eigenvalue of the fusion matrix $N_i$, so $|\chi_j(i)|\le\operatorname{FPdim}(i)$. Consequently,
\[
    \frac{\Delta}{q_j^2}
    =\sum_i\chi_j(i)\chi_j(\bar i)
    \le\sum_i|\chi_j(i)|\,|\chi_j(\bar i)|\le\Phi.
\]
Together with \eqref{bulk:eq:small-simple}, this gives
\begin{equation}\label{bulk:eq:exact-V}
    \min_j|j|^2=\frac{\Delta}{\Phi},\qquad V(\mathcal M)=\frac{\Phi}{\Delta}.
\end{equation}

\paragraph{The quantum-group families.}
For the $SU(2)$ family, take $h\ge3$ and $\gcd(m,2h)=1$ as in Section~\ref{sec:bulk}. The simple dimensions follow by expressing the quantum integers $[j+1]_q=(q^{j+1}-q^{-j-1})/(q-q^{-1})$ in trigonometric form. At $m=1$ these dimensions are positive and therefore agree with the Frobenius--Perron dimensions. Since multiplication by $m$ permutes the nonzero residues modulo $h$,
\[
    \sum_{j=1}^{h-1}\sin^2(jm\pi/h)=\frac h2.
\]
It follows that
\[
    \dim(\mathcal C_{h,m})=\frac{h}{2\sin^2(m\pi/h)},\qquad
    \operatorname{FPdim}(\mathcal C_{h,m})=\frac{h}{2\sin^2(\pi/h)}.
\]

First, examine the case of an odd $h$. Let $m$ be the odd integer among $(h-1)/2$ and $(h+1)/2$. $\gcd(m,2h)=1$ and $\sin^2(m\pi/h)=\cos^2(\pi/(2h))$. The identity $\sin(\pi/h)=2\sin(\pi/(2h))\cos(\pi/(2h))$ gives \eqref{bulk:eq:quartic}, with the asymptotic growth following from $\sin u=u+O(u^3)$.

As for the $SU(3)$ family, let $L$ be odd and $t_L$ the odd member of $(L-1)/2$ and $(L+1)/2$. Since $2t_L=L\pm1$, we have $\gcd(t_L,2L)=1$. As such, $t_L$ defines a cyclotomic Galois automorphism.

If the braiding formulas involve a primitive $6L$th root, we can choose an exponent coprime to $6L$ considering $t_L$ modulo $2L$. The resulting automorphism will extend to the algebraic closure and thus preserves the tensor identities.

For the adjoint $x=(1,1)$, the factors in \eqref{quant:qdim} give $\dim_t(x)=[2]_t[4]_t$. Let $\varepsilon=\pi/(2L)$ and let $t_L\pi/L=\pi/2+\eta\varepsilon$ for $\eta=\pm1$. Then
\[
    [2]_{t_L}=-2\eta\sin\varepsilon,\qquad
    [4]_{t_L}=4\eta\sin\varepsilon\cos(2\varepsilon),
\]
and hence
\[
    \dim_{t_L}(x)=-8\sin^2\varepsilon\cos(2\varepsilon).
\]
\eqref{quant:adjdim} is proved.

From
\[
    0<1-\cos(\pi/L)<\frac{\pi^2}{2L^2},\qquad 0<\cos(\pi/L)<1,
\]
we obtain $16c_L^2(1-c_L)^2<4\pi^4/L^4$, as in \eqref{quant:growth}.

%% file: 10-appendix-physical-lifts.tex
\section{Physical lifts, supports, and charged sectors}\label{app:lifts}

Consider linear maps $A:V\to\mathcal H$ and $B:W\to\mathcal H'$. A virtual operation $X:V\to W$ admits a physical lift $EA=BX$ when
\[
    \ker A\subseteq\ker(BX).
\]
This will make $E(Av)=BXv$ well-defined on $\ran A$ and we may extend it by zero on the orthogonal complement. When $A$ is injective on the specified support, the extension is simply $E=BXA^+$.

As for a fixed finite patch, let $A_c:S_c\to\Hh_c$ and $A_{ab}:S_{ab}\to\Hh_{ab}$ be injective PEPS maps on the charged virtual supports. Suppose there are virtual splitting and merging maps that satisfy
\[
    X_\alpha:S_c\to S_{ab},\qquad Y_\alpha:S_{ab}\to S_c,\qquad Y_\alpha X_\beta=\delta_{\alpha\beta}I_{S_c}.
\]
We will also assume these maps are independent of the separation of remote spectators. Corresponding physical maps are

\begin{equation}\label{charged:eq:fixedlifts}
    E_\alpha=A_{ab}X_\alpha A_c^+,\qquad R_\alpha^{\mathrm{phys}}=A_cY_\alpha A_{ab}^+.
\end{equation}
obeying
\[
    R_\alpha^{\mathrm{phys}}E_\beta A_c=A_cY_\alpha(A_{ab}^+A_{ab})X_\beta=\delta_{\alpha\beta}A_c.
\]
and have column norms such that
\begin{align*}
    S_p&\le\norm{A_{ab}}\norm{[X_1\ \cdots\ X_n]}\norm{A_c^+},\\
    R_p&\le\norm{A_c}\norm{[Y_1\ \cdots\ Y_n]^{\mathsf T}}\norm{A_{ab}^+}.
\end{align*}

Of course, the transported products must still satisfy the locality estimates in H1.

Let $\mathcal H_s$ denote the spectator Hilbert space. For any $\Psi\in\ran A_c\otimes\mathcal H_s$,
\[
    (R_\alpha^{\mathrm{phys}}E_\beta\otimes I_s)\Psi=\delta_{\alpha\beta}\Psi.
\]

If the virtual splitting obeys the gluing constraints of Corollary~\ref{charged:cor:fixedpatch}, the retractions restrict to the global charged codes. The same norm bounds hold after they tensor with $I_s$ and are restricted to these codes.

For an injective map, let $\kappa(A)=\norm{A}\norm{A^+}$.
\[
    S_pR_p\le\kappa(A_c)\kappa(A_{ab})\norm{[X_1\ \cdots\ X_n]}\norm{[Y_1\ \cdots\ Y_n]^{\mathsf T}}.
\]

\paragraph{Global inner products and PEPS supports.}\label{app:realization-background}

The commuting idempotents in \eqref{sn:Halg} are simultaneously diagonalizable, so the Hamiltonian has nonnegative integer spectrum. If $T$ maps a common eigenbasis to the original coordinates, then
\[
    G_{\mathrm{global}}=(T^{-1})^\dagger T^{-1}
\]
makes the idempotents orthogonal.

Because this metric might not arise from coherent elementary fusion metrics---or preserve locality---it does not imply vanishing associator defect.

It is worth noting that there is another way to obtain a Hermitian Hamiltonian from a string-net vector. Given a reduced density matrix $\rho_R$, define
\begin{equation}\label{sn:supportparent}
    h_R=I-\Pi_{\operatorname{supp}\rho_R},\qquad H_{\mathrm{supp}}=\sum_R h_R,\qquad H_{\mathrm{supp}}\lvert\psi\rangle=0.
\end{equation}

For charged sectors, the same condition is required on the corresponding charged supports.

The gap criterion of Kastoryano, Lucia, and Perez-Garcia relates approximate factorization of boundary states to a bulk gap \cite[Theorems 8 and 10]{KLP}. For noninjective PEPS, the boundary operators are restricted to specified supports.

Theorem IV.5 of \cite{Freedman} concerns at least five diagonal Galois-conjugated doubled Fibonacci anyons on the sphere, with a continuous uniform family of local normalizers and a uniformly gapped Lieb--Robinson Hamiltonian family as target.

\paragraph{Bulk sectors and physical charge measurements.}\label{app:bulk-operations}

Recall that a half-braiding on $x$ is a natural family $c_{x,a}:x\otimes a\to a\otimes x$ satisfying---strict monoidal notation---
\[
    c_{x,a\otimes b}=(\id_a\otimes c_{x,b})(c_{x,a}\otimes\id_b),\qquad c_{x,1}=\id_x.
\]
defining the Drinfeld center, for associators inserted in the non-strict case \cite[Section 7.13]{EGNO}.

Alternatively, bulk sectors can be described using idempotents of the annular string-net algebra with gluing multiplication is given by gluing annuli. Kirillov's circle-category equivalence identifies this construction with the Drinfeld center and relates it to spaces with labeled boundaries \cite[Theorems 6.4 and 7.3]{Kirillov}.

In an MPO realization, the virtual operators and fusion tensors will generate the corresponding annular algebra. The unitary construction admit a compatible $C^*$-structure \cite[Sections 4.1 and 5.1]{MPO}. Galois conjugation preserves the fusion equations \cite[Eqs. (6)--(9)]{Lootens}, but might not preserve the physical adjoint. Thus virtual charge idempotents do not need to correspond to physical orthogonal projections.